\documentclass[11pt]{article}
\usepackage{algorithm,algorithmicx,algpseudocode,amsmath,amssymb,amsthm}
\usepackage{bm,color,dsfont,epsfig,fullpage,graphicx,pifont,subfigure}
\usepackage{makecell}
\usepackage[toc,page]{appendix}
\usepackage{url,cite}
\usepackage{multirow}

\newtheorem{definition}{Definition}

\newtheorem{lemma}{Lemma}
\newtheorem{theorem}{Theorem}
\newtheorem{remark}{Remark}

\allowdisplaybreaks
\newcommand\numberthis{\addtocounter{equation}{1}\tag{\theequation}}

\def\la{\left\langle}
\def\ra{\right\rangle}
\def\lb{\left(}
\def\rb{\right)}
\def\lcb{\left\{}
\def\rcb{\right\}}
\def\ln{\left\|}
\def\rn{\right\|}
\def\lsb{\left[}
\def\rsb{\right]}
\def\C{\mathbb{C}}
\def\S{\mathcal{S}}
\def\P{\mathcal{P}}
\def\T{\mathcal{T}}
\def\H{\mathcal{H}}
\def\D{\mathcal{D}}
\def\G{\mathcal{G}}
\def\R{\mathcal{R}}
\def\E{\mathbb{E}}
\def\I{\mathcal{I}}
\def\BZ{\bm{Z}}
\def\BH{\bm{H}}
\def\BL{\bm{L}}
\def\BU{\bm{U}}
\def\BV{\bm{V}}

\def\BC{\bm{C}}
\def\BD{\bm{D}}
\def\BE{\bm{E}}
\def\BW{\bm{W}}
\def\BQ{\bm{Q}}
\def\BA{\bm{A}}
\def\BX{\bm{X}}
\def\BY{\bm{Y}}
\def\BQ{\bm{Q}}
\def\BR{\bm{R}}
\def\BS{\bm{\Sigma}}
\def\bx{\bm{x}}
\def\be{\bm{e}}
\def\bz{\bm{z}}
\def\by{\bm{y}}
\def\bb{\bm{b}}

\def\bf{\bm{f}}
\def\bt{\bm{\tau}}
\def\whm{{\widehat{m}}}
\def\vep{\varepsilon}
\DeclareMathOperator{\rank}{rank}
\DeclareMathOperator{\diag}{diag}
\DeclareMathOperator{\trace}{trace}
\begin{document}

\title{Fast and Provable Algorithms for Spectrally Sparse Signal Reconstruction via Low-Rank Hankel Matrix Completion}

\author{Jian-Feng Cai
\thanks{Department of Mathematics, Hong Kong University of Science and Technology, Clear Water Bay, Kowloon, Hong Kong SAR, China. Email: \texttt{jfcai@ust.hk}} 
\and Tianming Wang 
\thanks{Department of Mathematics, University of Iowa, Iowa City, Iowa, USA. Email: \texttt{tianming-wang@uiowa.edu}}
\and Ke Wei
\thanks{Department of Mathematics, University of California, Davis, California, USA. Email: \texttt{kewei@math.ucdavis.edu}}
}

\maketitle

\begin{abstract}
A spectrally sparse signal of order $r$ is a mixture of $r$ damped or undamped complex sinusoids. This paper investigates the problem of reconstructing spectrally sparse signals from a random subset of $n$ regular time domain samples, which can be reformulated as a low rank Hankel matrix completion problem. We introduce an iterative hard thresholding (IHT) algorithm and a fast iterative hard thresholding (FIHT) algorithm for efficient reconstruction of spectrally sparse signals via low rank Hankel matrix completion. Theoretical recovery guarantees have been established for FIHT, showing that $O(r^2\log^2(n))$ number of samples are sufficient for exact recovery with high probability. Empirical performance comparisons establish significant computational advantages for IHT and FIHT. In particular, 
numerical simulations on $3$D arrays demonstrate the capability of FIHT on handling large and high-dimensional real data.
\end{abstract}

{\textbf{Keywords.}}  Spectrally sparse signal,   low rank Hankel matrix completion,  iterative hard thresholding, composite hard thresholding operator
\section{Introduction}
Spectrally sparse signals arise frequently from various applications, ranging from magnetic resonance imaging \cite{MRI}, fluorescence microscopy \cite{Microscopy}, radar imaging \cite{Radar}, nuclear magnetic resonance (NMR) spectroscopy \cite{QMCCO:ACIE:15}, to analog-to-digital conversion \cite{Analog}. 
For ease of presentation, consider a one-dimensional ($1$-D) signal which is a weighted superposition of $r$ 
complex sinusoids with or without damping factors
\begin{equation}\label{eq:1D_cont_model}
x(t)=\sum_{k=1}^{r}d_ke^{\left(2\pi \imath f_k-\tau_k\right)t},
\end{equation}
where $\imath=\sqrt{-1}$,  $f_k\in[0,1)$ is the normalized frequency, $d_k\in\mathbb{C}$ is the corresponding complex amplitude, and $\tau_k\geq 0$ is the damping factor. Let $\bx=\begin{bmatrix}x_0,&\cdots,&x_{n-1}\end{bmatrix}^T\in\C^{n}$ be the discrete samples of $x(t)$ at $t\in\{0,\cdots,n-1\}$; that is,
\begin{equation}\label{eq:sig1D}
\bx=\begin{bmatrix}x(0), &\cdots,& x(n-1)\end{bmatrix}^T.
\end{equation} 
Under many circumstances of practical interests, $x(t)$ can  only be  sampled at a subset of times in $\{0,\cdots,n-1\}$ due to costly experiments \cite{QMCCO:ACIE:15}, hardware limitation \cite{Analog}, or other inevitable reasons. Consequently only partial entries of $\bx$ are known. Thus we need to reconstruct $\bx$ from its observed entries in these applications. Let $\Omega\subset\{0,\cdots,n-1\}$ with $|\Omega|=m$ be the collection of indices of the observed entries.  The reconstruction problem can be expressed as 
\begin{equation}\label{eq:sample}
\mbox{find}\quad \bx\quad\mbox{subject to}\quad\P_\Omega(\bx)=\sum_{a\in\Omega}x_a\be_a,
\end{equation}
where $\bm{e}_a$ is the $a$-th canonical basis of $\mathbb{C}^n$, and $\P_\Omega$ is a 
projection operator defined as 
\begin{align*}
\P_\Omega(\bz) = \sum_{a\in\Omega}\la \bz,\be_a\ra\be_a.
\end{align*}

Generally it is  impossible to reconstruct a vector from its partial entries since the unknown entries can take any values without violating the equality constraint in \eqref{eq:sample}. However, the theory of compressed sensing \cite{donoho2006cs,CS} and matrix completion \cite{candesrecht2009mc,rechtfazelparrilo2010nnm} suggests that signals with inherent simple structures can be uniquely determined from   a number of measurements  that is less than the size of the signal.  
In a spectrally sparse signal, the number of unknowns is at most $3r$, which is smaller than the length of the signal if $r\ll n$.
 Therefore it is possible to reconstruct $\bx$ from $\P_\Omega\bx$.

This paper exploits the low rank structure of the Hankel matrix constructed from $\bx$. 
Let $\H$ be a linear operator which maps a vector $\bz\in\C^{n}$ to a Hankel  matrix $\H\bz\in\C^{n_1\times n_2}$ with $n_1+n_2=n+1$ as follows
$$
[\mathcal{H}\bm{z}]_{ij}=z_{i+j}, \quad \forall ~i\in\{0,\ldots,n_1-1\},~j\in\{0,\ldots,n_2-1\}, 
$$
where vectors and matrices are indexed starting with zero, and $[\cdot]_{ij}$ denotes the $(i,j)$-th entry of a matrix.  
Define $y_k=e^{(2\pi\imath f_k-\tau_k)}$ for $k=1,\ldots,r$. Since $\bx$ is a spectrally sparse signal, the  Hankel matrix $\H\bx$ admits a Vandermonde decomposition 
$$
\mathcal{H}\bm{x}=\bm{E}_L\bm{D}\bm{E}_R^T,
$$
where
$$
\bm{E}_L=
\left[\begin{array}{cccc}
1         & 1         & \cdots & 1 \\
y_1       & y_2       & \cdots & y_r \\
\vdots    & \vdots    & \vdots & \vdots \\
y_1^{n_1-1} & y_2^{n_1-1} & \cdots & y_r^{n_1-1} \\
\end{array}\right],~
\bm{E}_R=
\left[\begin{array}{cccc}
1         & 1         & \cdots & 1 \\
y_1       & y_2       & \cdots & y_r \\
\vdots    & \vdots    & \vdots & \vdots \\
y_1^{n_2-1} & y_2^{n_2-1} & \cdots & y_r^{n_2-1} \\
\end{array}\right]
$$
and
$\bm{D}$ is a diagonal matrix whose diagonal entries are  $d_1,\dots,d_r$. If all $y_k$'s are distinct  and $r\leq\min(n_1,n_2)$, $\BE_L$ and $\BE_R$ are both full rank matrices.
Therefore $\rank(\H\bx)=r$ when all $d_k$'s are non-zeros. Since $\H$ is injective, the reconstruction of $\bx$ from $\P_\Omega(\bx)$ is equivalent to the reconstruction of $\H\bx$ from partial revealed anti-diagonals that corresponds to the known entries of $\bx$. 
With a slight abuse of notation we also use $\P_\Omega$ to denote the projection of a matrix $\BZ\in\C^{n_1\times n_2}$ onto the subspace  determined by a subset of an orthonormal basis of Hankel matrices; that is,
\begin{align*}
\P_\Omega(\BZ) = \sum_{a\in\Omega}\la\BZ,\BH_a\ra \BH_a,
\end{align*}
where the set of matrices
\begin{align*}\lcb\BH_a=\frac{1}{\sqrt{w_a}}\H\be_a~|~ w_a=\#\lcb (i,j)~|~ i+j=a,~ 0\leq i\leq n_1-1,~0\leq j\leq n_2-1\rcb\rcb_{a=0}^{n-1}\numberthis\label{eq:Ha}
\end{align*} 
 forms an orthonormal basis of $n_1\times n_2$ Hankel matrices.
To reconstruct $\H\bx$, we seek the lowest rank Hankel matrix consistent with the revealed anti-diagonals by solving the
following {\em low rank Hankel matrix completion} problem
\begin{align}\label{eq:min_hankel_rank}
\min_{\bz}\rank(\H\bz)\quad\mbox{subject to}\quad \P_\Omega(\H\bz)=\P_\Omega(\H\bx).
\end{align}

{\em In this paper, we first develop an iterative hard thresholding  (IHT) algorithm to reconstruct spectrally sparse signals via  low rank Hankel matrix completion. Then the algorithm is further accelerated by  applying subspace projections to reduce the high per iteration computational complexity of the singular value decomposition, which leads to a fast iterative hard thresholding (FIHT) algorithm. Moreover, FIHT has been proved to be able to converge linearly to the unknown signal with high probability  if the number of revealed entries is of the order $O(r^2\log^2(n))$ and the algorithm is properly initialized.}
\subsection{Overview of Related Work}
In a paper that is mostly related to our work, Chen and Chi \cite{Chi} study nuclear norm minimization for  the low rank Hankel matrix completion problem, where $\rank(\H\bz)$ in \eqref{eq:min_hankel_rank} is replaced by the nuclear norm of $\H\bz$.
The authors show that $O(r\log^4(n))$ randomly selected samples are sufficient to guarantee 
exact recovery of spectrally sparse signals with high probability under some mild incoherence conditions. Theoretical recovery guarantees are also established in \cite{Chi} for robustness of nuclear norm minimization under bounded additive noise and sparse outliers. Nuclear norm minimization for the low rank Hankel matrix reconstruction problem under the random Gaussian sampling model is investigated in \cite{CQXY:ACHA:16}.

In a different direction, the sparsity of  $\bx$ in the frequency domain can be utilized to develop reconstruction algorithms. When there is no damping, i.e., $\tau_r=0$ for all $r$, one may discretize the frequency domain $[0,1)$ by a uniform grid and then use conventional compressed sensing \cite{CS,donoho2006cs} to  estimate the spectrum of $\bx$.  However, in many applications the true frequencies are continuous-valued. The discretization error will cause the so-called basis mismatch \cite{Mismatch}, resulting in the loss of sparsity of the signal under the discrete Fourier transform and consequently the degradation in recovery performance.  In \cite{Tang}, Tang et al. exploit the sparsity of $\bx$ in a continuous way via the atomic norm. They show that exact recovery with high probability  can be established from $O(r\log(r)\log(n))$ random time domain samples, provided that the complex amplitudes of $\bx$ have uniformly distributed random phases and the minimum wrap-around distance between its frequencies is at least $4/n$.

The methods developed in \cite{Chi} and \cite{Tang} utilize convex relaxation and are theoretically guaranteed to work. However, the common drawback of these otherwise very appealing convex optimization approaches is the high computational complexity of solving the equivalent semi-definite programming (SDP).  In \cite{PWGD}, Cai et al. develop a  fast non-convex algorithm for low rank Hankel matrix completion by minimizing the distance between low rank matrices and Hankel matrices with partial known anti-diagonals. The proposed algorithm   has been proved to  be able to converge to a critical point of the cost function. An accelerated variant has also been developed in \cite{PWGD} using Nesterov's memory technique as inspired by FISTA \cite{FISTA}.
\subsection{Notation and Organization of the Paper} 
The rest of the paper is organized as follows.  We first summarize the notation used throughout
this paper in the remainder of this section. The  IHT and FIHT algorithms are presented at the beginning of Sec.~\ref{sec:alg}, followed by the implementation details, theoretical recovery guarantees, extension to higher dimensions and connections to tight frame analysis sparsity in compressed sensing. Numerical evaluations in Sec.~\ref{sec:NumExp} demonstrate the efficiency of the proposed algorithms and their applicability for real applications. The proofs of the main results are presented in Sec.~\ref{sec:proofs} and Sec.~\ref{sec:conclusion} concludes this paper with future research directions.

Throughout this paper, we denote vectors by bold lowercase letters and matrices by bold uppercase letters.  Vectors and matrices are indexed starting with zero. The individual entries of vectors and matrices are denoted by normal font. For any matrix $\bm{Z}$, $\|\bm{Z}\|$, $\|\bm{Z}\|_F$, $\|\bm{Z}\|_\infty$ respectively denote its spectral norm, Frobenius norm, and the maximum magnitude of its entries respectively.  The $i$-th row and $j$-th column of a matrix  $\bm{Z}$ are denoted by $\bm{Z}^{(i,:)}$ and $\bm{Z}^{(:,j)}$ respectively. The transpose of vectors and matrices is denoted by $\bz^T$ and $\BZ^T$, while their conjugate transpose is denoted by $\bz^*$ and $\BZ^*$. The inner product of two matrices is defined as $\la \BZ_1,\BZ_2\ra=\trace(\BZ_2^*\BZ_1)$, and the inner product of two vectors is given by $\la\bz_1,\bz_2\ra=\bz_2^*\bz_1$. For a natural number $n$, we denote the set $\lcb 0,\cdots, n-1\rcb$ by $\lsb n\rsb$.

Operators are denoted by calligraphic letters. In particular,  $\I$ denotes the identity operator and $\H$ denotes the Hankel operator which maps an $n$-dimensional vector to an $n_1\times n_2$ Hankel matrix with $n_1+n_2=n+1$. The ratio $c_s$ is defined as $c_s=\max\{\frac{n}{n_1},\frac{n}{n_2}\}$.
We denote  the adjoint of $\H$ by $\H^*$, which is a linear operator from $n_1\times n_2$ matrices to $n$-dimensional vectors. For any matrix $\BZ\in\C^{n_1\times n_2}$, simple calculation reveals that $\H^*\BZ =\lcb\sum_{i+j=a}Z_{ij}\rcb_{a=0}^{n-1} $.  Define $\D^2=\H^*\H$. Then it is a diagonal operator from vectors to vectors of the form $\D^2\bz = \lcb w_az_a\rcb_{a=0}^{n-1}$, where $w_a$ defined in \eqref{eq:Ha} is the number of elements in $a$-th anti-diagonal of an $n_1\times n_2$ matrix.
The Moore-Penrose pseudoinverse of $\H$ is given by $\H^\dag=\D^{-2}\H^*$ which satisfies $\H^\dag\H=\I$. Finally, we use $C$ to denote a universal numerical constant whose value may change from line to line.

 \section{Algorithms and Theoretical Results}\label{sec:alg}
\subsection{Algorithms}
We present our first reconstruction algorithm in Alg.~\ref{IHT-GD}, which is an iterative hard thresholding  algorithm for the following  reformulation of \eqref{eq:min_hankel_rank},
\begin{align}\label{eq:obj_fun}
\min_{\bm{z}}\la\bz-\bx,\P_\Omega(\bz-\bx)\ra\quad\mbox{subject to}\quad \rank(\mathcal{H}\bm{z})=r.
\end{align}
In each iteration of IHT, the current estimate $\bx_l$ is first updated along the gradient descent direction  
under the Wirtinger calculus with the stepsize $p^{-1}=\frac{n}{m}$.  Then the Hankel matrix corresponding to the update is formed via the application of the linear operator $\H$, followed by an SVD truncation to its nearest rank $r$ approximation. The hard thresholding operator $\T_r(\cdot)$ in Step $3$ of Alg.~\ref{IHT-GD} is defined as 
\begin{align*}
\mathcal{T}_r(\bm{Z})=\sum_{k=1}^r\sigma_r\bm{u}_k\bm{v}_k^*,
\quad\mbox{where }\bm{Z}=\sum_{k=1}^{\min(n_1,n_2)}\sigma_k\bm{u}_k\bm{v}_k^*\mbox{ is an SVD with }\sigma_1\geq\sigma_2\geq\ldots\geq\sigma_{\min(n_1,n_2)}.
\end{align*}
Finally the new estimate $\bx_{l+1}$ is obtained via the application of $\H^\dag$ on the low rank matrix $\BL_{l+1}$.
\begin{algorithm}[htp]
\caption{Iterative Hard Thresholding (IHT)}
\label{IHT-GD}
\begin{algorithmic} 
\Statex \textbf{Initialize} $\bm{L}_{0}$ and \textbf{Set} $\bm{x}_{0}=\mathcal{H}^{\dag}\bm{L}_{0}$
\For{$l=0,1,\cdots$}\\
1.  $\bm{g}_l=\mathcal{P}_{\Omega}(\bm{x}-\bm{x}_{l})$\\
2.  $\bm{W}_{l}=\mathcal{H}(\bm{x}_{l}+p^{-1}\bm{g}_l)$\\
3.  $\bm{L}_{l+1}=\mathcal{T}_{r}\left(\bm{W}_{l}\right)$\\
4.  $\bm{x}_{l+1}=\mathcal{H}^{\dag}\bm{L}_{l+1}$
\EndFor
\end{algorithmic}
\end{algorithm}

\begin{algorithm}[htp]
\caption{Fast Iterative Hard Thresholding (FIHT)}
\label{Alg-GD}
\begin{algorithmic} 
\Statex \textbf{Initialize} $\bm{L}_{0}$ and \textbf{Set} $\bm{x}_{0}=\mathcal{H}^{\dag}\bm{L}_{0}$
\For{$l=0,1,\cdots$}\\
1. $\bm{g}_l=\mathcal{P}_{\Omega}(\bm{x}-\bm{x}_{l})$\\
2. $\bm{W}_{l}=\P_{\S_l}\mathcal{H}(\bm{x}_{l}+p^{-1}\bm{g}_l)$\\
3. $\bm{L}_{l+1}=\mathcal{T}_{r}\left(\bm{W}_{l}\right)$\\
4. $\bm{x}_{l+1}=\mathcal{H}^{\dag}\bm{L}_{l+1}$
\EndFor
\end{algorithmic}
\end{algorithm}

Empirically, IHT can achieve linear convergence rate  as demonstrated in Sec.~\ref{sec:nu_speed}. However, it requires to compute the truncated SVD of an $n_1\times n_2$ matrix in each iteration. Though there are fast SVD solvers \cite{PROPACK,xu2008fast}, it is still computationally expensive when $n$ ($=n_1+n_2-1$) is large. 
To improve the computational efficiency we propose to project the Hankel matrix $\mathcal{H}(\bm{x}_{l}+p^{-1}\bm{g}_l)$ onto a low dimensional subspace $\S_l$ before truncating it to the best rank $r$ approximation. 
The fast iterative hard thresholding algorithm equipped with an extra subspace  projection step is presented in Alg.~\ref{Alg-GD}, 
where $\P_{\S_l}(\cdot)$ denotes the projection of $n_1\times n_2$ matrices onto the subspace $\S_l$.
Inspired by the Riemannian optimization algorithms
for  low rank matrix completion \cite{Recovery,Completion,bart2012riemannian}, $\S_l$ is selected to be the direct sum of  the column and  row spaces of  $\BL_l$,
\begin{align}\label{eq:subspace}
\mathcal{S}_l=\{\bm{U}_l\bm{B}+\bm{C}\bm{V}_l^{*}~|~\bm{B}\in\mathbb{C}^{r\times n_2},~\bm{C}\in\mathbb{C}^{n_1\times r}\},
\end{align}
where $\BU_l\in\C^{n_1\times r}$ and $\BV_l\in\C^{n_2\times r}$ are the left and right singular vectors of $\BL_l$. 
The subspace $\S_l$ defined in \eqref{eq:subspace} can be geometrically interpreted as the tangent space of the embedded rank $r$ matrix manifold at $\BL_l$ \cite{bart2012riemannian}. For any matrix $\BZ\in\C^{n_1\times n_2}$, the projection of $\BZ$ onto $\S_l$ is given by 
\begin{align*}
\P_{\S_l}\lb\BZ\rb = \BU_l\BU_l^*\BZ+\BZ\BV_l\BV_l^*-\BU_l\BU_l^*\BZ\BV_l\BV_l^*.
\end{align*}

Iterative hard thresholding  is a family of simple yet efficient algorithms for compressed sensing \cite{bludav2009iht, blumensathdavies2010niht, foucart2011htp,CGIHT} and matrix completion \cite{tw2012nihtmc,jmd2010svp,goldfarbma2011fpca}, where in compressed sensing signals of interest are sparse and in matrix completion signals of interest are low rank. However, the signal of interest  in this paper is neither sparse nor low rank itself, but instead the Hankel matrix corresponding to  the signal is low rank. Therefore Algs.~\ref{IHT-GD} and \ref{Alg-GD} alternate between the vector space and the matrix space and this alternating structure does not exist in typical iterative hard thresholding algorithms for compressed sensing and matrix completion.
\subsection{Implementation and Computational Complexity} 
We focus on the implementation details of FIHT and show that the SVD of $\BW_l$ in the third step of Alg.~\ref{Alg-GD} can be computed using $O(r^3)$ floating point operations (flops) owing to the low rank structure of the matrices in $\S_l$. The implementation of IHT is similar to that of FIHT, except that the computation of the SVD of $\BW_l$ generally requires $O(n^3)$ flops.

Assume the rank $r$ matrix $\BL_l$ is stored by its SVD $\BL_l=\BU_l\BS_l\BV_l^*$ in each iteration. Then, 
\begin{align*}
\bx_l = \H^\dag\BL_l = \D^{-2}\H^*\BL_l = \D^{-2} \sum_{k=1}^r\BS_l^{(k,k)}\H^*\lb\BU_l^{(:,k)}\lb\BV_l^{(:,k)}\rb^*\rb,
\end{align*}
where $\H^*\lb\BU_l^{(:,k)}\lb\BV_l^{(:,k)}\rb^*\rb$ can be computed via fast convolution by noting that 
\begin{align*}
\lsb \H^*\lb\BU_l^{(:,k)}\lb\BV_l^{(:,k)}\rb^*\rb\rsb_a = \sum_{i+j=a}\BU_l^{(i,k)}\overline{\BV_l}^{(j,k)},\quad a = 0,\cdots,n-1.
\end{align*}
Therefore computing the last step of Alg.~\ref{Alg-GD} costs $O(rn\log(n))$ flops. 

We distinguish two cases regarding to the 
computations of $\BW_l$ and its SVD.

{\em Case $1$: $n_1\ne n_2$.} Let $\BH_l=\H\lb\bm{x}_{l}+p^{-1}\bm{g}_l\rb$. The intermediate matrix $\BW_l$ is stored by the following decomposition
\begin{equation*}
\begin{split}
\bm{W}_l=\mathcal{P}_{\mathcal{S}_l}\bm{H}_l
&=\bm{U}_l\bm{U}_l^*\bm{H}_l+\bm{H}_l\bm{V}_l\bm{V}_l^*-\bm{U}_l\bm{U}_l^*\bm{H}_l\bm{V}_l\bm{V}_l^*\cr
&=\bm{U}_l\underbrace{\bm{U}_l^*\bm{H}_l\bm{V}_l}_{\bm{C}\in \mathbb{C}^{r\times r}}\bm{V}_l^*+
\bm{U}_l\underbrace{\bm{U}_l^*\bm{H}_l(\bm{I}-\bm{V}_l\bm{V}_l^*)}_{\bm{X}^{*}\in \mathbb{C}^{r\times n_2}}+
\underbrace{(\bm{I}-\bm{U}_l\bm{U}_l^*)\bm{H}_l\bm{V}_l}_{\bm{Y}\in \mathbb{C}^{n_1\times r}}\bm{V}_l^*
\cr
&=\bm{U}_l\bm{C}\bm{V}_l^*+\bm{U}_l\bm{X}^*+\bm{Y}\bm{V}_l^*.
\end{split}
\end{equation*}
Note that  $\BH_l^*\BU_l$ and $\BH_l\BV_l$ in $\BC$, $\BX$ and $\BY$ can be computed using $r$ fast Hankel matrix-vector multiplications without forming $\BH_l$ explicitly, which requires $O(rn\log(n))$ flops.  Therefore the total 
computational cost for computing $\BC$, $\BX$ and $\BY$ is $O(r^2n+rn\log(n))$ flops.

Let $\bm{X}=\bm{Q}_1\bm{R}_1$ and $\bm{Y}=\bm{Q}_2\bm{R}_2$ respectively be the QR factorizations of $\BX$ and $\BY$. Then $\BQ_1\perp\BV_l$, $\BQ_2\perp\BU_l$ and $\BW_l$ can be  rewritten as 
\begin{align*}
\bm{W}_l=\bm{U}_l\bm{C}\bm{V}_l^*+\bm{U}_l\bm{R}_1^*\bm{Q}_1^*+\bm{Q}_2\bm{R}_2\bm{V}_l^*
=\begin{bmatrix}
\BU_l & \BQ_2
\end{bmatrix}
\begin{bmatrix}
\BC & \BR_1^*\\
\BR_2 & \bm{0}
\end{bmatrix}
\begin{bmatrix}
\BV_l & \BQ_1
\end{bmatrix}^*.
\end{align*}
Suppose the SVD of the middle $2r\times 2r$ matrix is given by 
\begin{align*}
\begin{bmatrix}
\BC & \BR_1^*\\
\BR_2 & \bm{0}
\end{bmatrix}=\BU_c\BS_c\BV_c^*.
\end{align*}
Then SVD of $\BW_l$ can be computed as 
\begin{align*}
\BW_l=\lb\begin{bmatrix}
\BU_l & \BQ_2
\end{bmatrix}\BU_c \rb\BS_c\lb \begin{bmatrix}
\BV_l & \BQ_1
\end{bmatrix}\BV_c\rb^*.
\end{align*}
Thus computing the SVD of $\BW_l$ requires  $O(r^2n+r^3)$ flops.
 
{\em Case 2: $n_1=n_2$.} In this case, $\BH_l$ is a square and symmetric matrix (but not Hermitian). Assume $\BL_l$ is also symmetric which can be achieved when $l=0$. Then $\BL_l$ admits a Takagi factorization $\BL_l=\BU_l\BS_l\BU_l^T$, which is also the SVD of $\BL_l$ \cite{xu2008fast}. So 
\begin{equation*}
\begin{split}
\bm{W}_l=\mathcal{P}_{\mathcal{S}_l}\lb\bm{H}_l\rb
&=\bm{U}_l\bm{U}_l^*\bm{H}_l+\bm{H}_l\overline{\bm{U}_l}\bm{U}_l^T-\bm{U}_l\bm{U}_l^*\bm{H}_l\overline{\bm{U}_l}\bm{U}_l^T\cr
&=\bm{U}_l\underbrace{\bm{U}_l^*\bm{H}_l\overline{\bm{U}_l}}_{\bm{C}\in \mathbb{C}^{r\times r}}\bm{U}_l^T+
\bm{U}_l\underbrace{\bm{U}_l^*\bm{H}_l(\bm{I}-\overline{\bm{U}_l}\bm{U}_l^T)}_{\bm{X}^{T}\in \mathbb{C}^{r\times n_1}}+
\underbrace{(\bm{I}-\bm{U}_l\bm{U}_l^*)\bm{H}_l\overline{\bm{U}_l}}_{\bm{X}\in \mathbb{C}^{n_1\times r}}\bm{U}_l^T
\cr
&=\bm{U}_l\bm{C}\bm{U}_l^T+\bm{U}_l\bm{X}^T+\bm{X}\bm{U}_l^T
\end{split}
\end{equation*}
is also a symmetric matrix and nearly half of the computational costs will be saved compared with the non-square case.

Let $\bm{X}=\bm{Q}\bm{R}$ be the QR factorization of $\BX$. Then $\BQ\perp\BU$ and 
\begin{align*}
\bm{W}_l=\bm{U}_l\bm{C}\bm{U}_l^T+\bm{U}_l\bm{R}^T\bm{Q}^T+\bm{Q}\bm{R}\bm{U}_l^T=\begin{bmatrix}
\BU_l&\BQ
\end{bmatrix}
\begin{bmatrix}
\BC & \BR^T\\\BR & \bm{0}
\end{bmatrix}
\begin{bmatrix}
\BU_l&\BQ
\end{bmatrix}^T.
\end{align*}
This, together with  the Takagi factorization (also the SVD) of the middle $2r\times 2r$ matrix 
$$
\begin{bmatrix}
\bm{C}   & \bm{R}^T \\
\bm{R} & \bm{0}
\end{bmatrix}=\bm{U}_c\bm{\Sigma}_c\bm{U}_c^T,
$$
gives the Takagi factorization (also the SVD) of $\bm{W}_l$
\begin{equation*}
\begin{split}
\bm{W}_l=\left(\begin{bmatrix}
\bm{U}_l & \bm{Q}\\
\end{bmatrix}
\bm{U}_c\right)\bm{\Sigma}_c\left(
\begin{bmatrix}
\bm{U}_l & \bm{Q} \\
\end{bmatrix}\bm{U}_c\right)^T.
\end{split}
\end{equation*}
Moreover, $\BL_{l+1}$ remains symmetric and admits a  Takagi factorization as the best rank $r$ approximation of $\BW_l$. 

In summary, the leading order per iteration computational cost of FIHT is $O(r^2n+rn\log (n)+r^3)$ flops, which can be further reduced by exploring the symmetric structure of matrices when $n_1=n_2$. In addition,  the largest matrices that need to be stored are the singular vector matrices of $\BW_l$. Therefore, FIHT requires only $O(rn)$ memory.
\subsection{Initializations and Recovery Guarantees}\label{sec:theory}
In this section, we present theoretical recovery guarantees for FIHT (Alg.~\ref{Alg-GD}). The guarantee analysis relies on restricted isometry properties of $\P_\Omega$ which cannot be established for IHT (Alg.~\ref{IHT-GD}). Moreover, numerical simulations in Sec.~\ref{sec:NumExp} suggest that while FIHT and IHT both have linear convergence rate, FIHT can be sufficiently faster  due to the low per iteration computational cost.

Let $\Omega=\{a_k~|~k=1,\ldots,m\}$. We consider the {\em sampling with replacement model} for $\Omega$; that is each index $a_k$ is drawn independently and uniformly from $\lcb 0,\cdots,n-1\rcb$.  Recall that we use $\P_\Omega(\cdot)$ to represent the projection of vectors onto a subset of the canonical basis of $\C^n$, i.e., 
\begin{align*}
\P_\Omega(\bz) = \sum_{k=1}^m \la \bz,\be_{a_k}\ra\be_{a_k},\quad\forall\bz\in\C^n
\end{align*}
as well as the projection of matrices onto a subset of an orthonormal basis of Hankel matrices, i.e.,
\begin{align*}
\P_\Omega(\BZ) = \sum_{k=1}^m\la\BZ,\BH_{a_k}\ra\BH_{a_k}, ~\quad\forall\BZ\in\C^{n_1\times n_2}
\end{align*}
since they are corresponding to each other and the context will make their distinction clear.  The key insight in  matrix completion suggests that in order to achieve successful low rank Hankel matrix completion, it requires the singular vectors of  the underlying Hankel matrix $\H\bx$ are not aligned with the orthonormal basis $\lcb\BH_a\rcb_{a=0}^{n-1}$. This can be guaranteed if the smallest singular values of the  left matrix $\BE_L$ and the right matrix $\BE_R$ in the Vandermonde decomposition of $\H\bx$ are bounded away from zero.
\begin{definition}\label{def:incoherence}
The rank $r$ Hankel matrix $\mathcal{H}\bm{x}$ with the Vandermonde decomposition $\H\bx=\BE_L\BD\BE_R^T$ is said to be  $\mu_0$-incoherent if there exists a numerical constant $\mu_0>0$ such that
$$
\sigma_{\min}(\bm{E}_L^*\bm{E}_L)\geq \frac{n_1}{\mu_0},
~
\sigma_{\min}(\bm{E}_R^*\bm{E}_R)\geq \frac{n_2}{\mu_0}.
$$
\end{definition} 
This incoherence property was introduced in \cite{Chi} and is crucial to our proofs.  Moreover we know from \cite[Thm.~2]{MUSIC} that, in the undamping case, if the minimum wrap-around distance between the frequencies is greater than about $\frac{2}{n}$, this property can be satisfied.  Let $\H\bx=\BU\BS\BV^*$ be the reduced SVD of $\H\bx$ and $\P_{\BU}(\cdot)$ and $\P_{\BV}(\cdot)$ respectively be the orthogonal projections onto the subspaces spanned by $\BU$ and $\BV$. 
The following lemma follows directly from Def.~\ref{def:incoherence}.
\begin{lemma}\label{lem:U_V}
Let $\H\bx=\BU\BS\BV^*=\BE_L\BD\BE_R^T$. Assume $\H\bx$ is $\mu_0$ incoherent and define  $c_s=\max\lcb\frac{n}{n_1},\frac{n}{n_2}\rcb$. Then
\begin{align*}
\ln\BU^{(i,:)}\rn^2\leq\frac{\mu_0c_sr}{n}\quad&\mbox{and}\quad\ln\BV^{(j,:)}\rn^2\leq\frac{\mu_0c_sr}{n}\numberthis\label{eq:row_norm_U_V},\\
\ln\P_{\BU}(\BH_a)\rn_F^2\leq\frac{\mu_0c_sr}{n}\quad&\mbox{and}\quad \ln\P_{\BV}(\BH_a)\rn_F^2\leq \frac{\mu_0c_sr}{n},\numberthis\label{eq:proj_norm_U_V}
\end{align*}
\end{lemma}
\begin{proof}
The proof of  \eqref{eq:proj_norm_U_V} can be found in \cite{Chi}. We include the proof here to be self-contained. We only prove the left inequalities of \eqref{eq:row_norm_U_V} and \eqref{eq:proj_norm_U_V} as the right ones can be similarly established. Since $\BU\in\C^{n_1\times r}$ and $\BE_l\in\C^{n_1\times r}$ spans the same subspace and $\BU$ is orthogonal, there exists an orthonormal matrix $\BQ\in\C^{r\times r}$ such that  $\BU=\BE_L(\BE_L^*\BE_L)^{-1/2}\BQ$. So 
\begin{align*}
\ln\BU^{(i,:)}\rn^2=\ln \be_i^*\BE_L(\BE_L^*\BE_L)^{-1/2}\rn^2\leq \ln \be_i^*\BE_L\rn^2\ln (\BE_L^*\BE_L)^{-1}\rn\leq\frac{\mu_0r}{n_1}\leq \frac{\mu_0c_sr}{n}
\end{align*}
 and 
 \begin{align*}
 \ln\P_{\BU}(\BH_a)\rn_F^2 &= \ln\BU\BU^*\BH_a\rn_F^2=\ln \BE_L(\BE_L^*\BE_L)^{-1}\BE_L^*\BH_a\rn_F^2\leq\frac{\ln \BE_L^*\BH_a\rn_F^2}{\sigma_{\min}(\bm{E}_L^*\bm{E}_L)}\leq\frac{\mu_0r}{n_1}\leq\frac{\mu_0c_sr}{n},
 \end{align*}
 where we have  used the fact that $\BH_a$ only has $w_a$ nonzero entries of magnitude $1/\sqrt{w_a}$  in its $a$-th anti-diagonal and the magnitudes of the entries of $\BE_L$ is bounded above by one for both the damped and undampled case.
\end{proof}

As is typical in non-convex optimization, the theoretical recovery guarantees of FIHT are closely related to the initial guess. We will discuss two initialization strategies and the corresponding recovery guarantees for FIHT.
The proofs of the lemmas and theorems in Secs.~\ref{sec:one_step} and \ref{sec:resampling} will be provided in Sec.~\ref{sec:proofs}.
\subsubsection{Initialization via One Step Hard Thresholding}\label{sec:one_step}
Our first initial guess is $\BL_0=p^{-1}\T_r(\H\P_\Omega(\bx))$, which is obtained by truncating the Hankel matrix constructed from the observed entries of $\bx$. The following lemma which is of independent interest  bounds the deviation of $\BL_0$ from $\H\bx$.
\begin{lemma}\label{lem:initial}
Assume $\mathcal{H}\bm{x}$ is $\mu_0$-incoherent. Then there exists a universal constant $C>0$ such that
$$
\|\bm{L}_0-\mathcal{H}\bm{x}\| \leq C\sqrt{\frac{\mu_0c_s r\log(n)}{m}}\|\mathcal{H}\bm{x}\|
$$ 
 with probability at least $1-n^{-2}$.
\end{lemma}
It follows from Lem.~\ref{lem:initial} that, if $m$ is sufficiently large and in the order of $r\log(n)$, the spectral norm distance between $\BL_0$ and $\H\bx$ can be less than any arbitrarily small constant. The following theoretical recovery guarantee can be established for FIHT based on this lemma.
\begin{theorem}[Guarantee I]\label{thm:IHT}
Assume $\mathcal{H}\bm{x}$ is $\mu_0$-incoherent. Let $0<\vep_0<\frac{1}{10}$ be a numerical constant and $\nu=10\vep_0<1$. Then with probability at least $1-3n^{-2}$, the iterates generated by FIHT (Alg.~\ref{Alg-GD}) with the initial guess $\bm{L}_0=p^{-1}\T_r(\H\P_\Omega(\bx))$ satisfy
$$
\|\bm{x}_{l}-\bm{x}\|\leq \nu^l\|\bm{L}_0-\mathcal{H}\bm{x}\|_F,
$$
provided
$$
m\geq C\max\left\{\vep_0^{-2}\mu_0c_s,(1+\vep_0)\vep_0^{-1}\mu_0^{1/2}c_s^{1/2}\right\}\kappa r n^{1/2} 
\log^{3/2}(n)
$$ 
for some universal constant $C>0$, where $\kappa=\frac{\sigma_{\max}(\H\bx)}{\sigma_{\min}(\H\bx)}$ denotes the condition number of $\H\bx$.
\end{theorem}
\begin{remark}{\normalfont
Since $\H\bx=\BE_L\BD\BE_R^T$, we have 
\begin{align*}
\kappa\leq\frac{\sigma_{\max}(\BE_L)}{\sigma_{\min}(\BE_L)}\cdot\frac{\max_{k}|d_k|}{\min_{k}|d_k|}\cdot\frac{\sigma_{\max}(\BE_R)}{\sigma_{\min}(\BE_R)}.
\end{align*}
It follows from \cite[Thm.~2]{MUSIC} that $\sigma_{\max}(\BE_L)$ (resp. $\sigma_{\max}(\BE_R)$) and $\sigma_{\min}(\BE_L)$ (resp. $\sigma_{\min}(\BE_R)$)  are both proportional to $\sqrt{n_1}$ (resp. $\sqrt{n_2}$) when the frequencies of $\bx$ are well separated. Thus the condition number of $\H\bx$ is essentially proportional to the dynamical range $\max_{k}|d_k|/\min_{k}|d_k|$.

Since the number of measurements required in Thm.~\ref{thm:IHT} is proportional to $c_s=\max\lcb\frac{n}{n_1},\frac{n}{n_2}\rcb$ and $n_1+n_2-1=n$, it makes sense to construct a nearly square Hankel matrix to recover spectrally sparse signals via low rank Hankel matrix completion.
}\end{remark}
\subsubsection{Initialization via Resampling and Trimming}\label{sec:resampling}
The sampling complexity in Thm.~\ref{thm:IHT} depends on $\sqrt{n}$ which is no desirable since the degrees of freedom in a spectrally sparse signal is only proportional to $r$. To eliminate the dependence on $\sqrt{n}$, we 
investigate another initialization procedure which is described in Alg.~\ref{Resampling}. 
\begin{algorithm}[htp]
\caption{Initialization via Resampled FIHT and Trimming}
\label{Resampling}
\begin{algorithmic}[]
\State \textbf{Partition} $\Omega$ into $L+1$ disjoint sets $\Omega_0,\cdots,\Omega_{L}$ of equal size $\widehat{m}$, let $\widehat{p}=\frac{\widehat{m}}{n}$.
\State \textbf{Set} $\widetilde{\bm{L}}_0=\mathcal{T}_r\lb\widehat{p}^{-1}\mathcal{H}\mathcal{P}_{\Omega_0}\lb\bm{x}\rb\rb$, 
\For{$l=0,\cdots,L-1$}
\State 1. $\widehat{\bm{L}}_{l}=\mathrm{Trim}_{\mu_0}(\widetilde{\bm{L}}_{l})$
\State 2. $\widehat{\bm{x}}_l=\mathcal{H}^{\dag}\widehat{\bm{L}}_l$
\State 3. $\widetilde{\bm{L}}_{l+1}=\mathcal{T}_{r}\mathcal{P}_{\widehat{\mathcal{S}}_l}\mathcal{H}\left(\widehat{\bm{x}}_l+\widehat{p}^{-1}\mathcal{P}_{\Omega_{l+1}}\left(\bm{x}-\widehat{\bm{x}}_l\right)\right)$
\EndFor
\end{algorithmic}
\end{algorithm}

\begin{algorithm}[htp]
\caption{$\mathrm{Trim}_{\mu}$}
\label{Trimming}
\begin{algorithmic}[]
\State \textbf{Input:} $\widetilde{\bm{L}}_{l+1}=\widetilde{\bm{U}}_{l+1}\widetilde{\bm{\Sigma}}_{l+1}\widetilde{\bm{V}}_{l+1}^*$
\State \textbf{Output:} $\widehat{\bm{L}}_{l+1}=\widehat{\bm{A}}_{l+1}\widetilde{\bm{\Sigma}}_{l+1}\widehat{\bm{B}}_{l+1}^*$, where
$$
\widehat{\bm{A}}_{l+1}^{(i,:)}=\frac{\widetilde{\bm{U}}_{l+1}^{(i,:)}}{\left\|\widetilde{\bm{U}}_{l+1}^{(i,:)}\right\|} \min\left\{\left\|\widetilde{\bm{U}}_{l+1}^{(i,:)}\right\|,\sqrt{\frac{\mu c_s r}{n}}\right\},
\quad
\widehat{\bm{B}}_{l+1}^{(i,:)}=\frac{\widetilde{\bm{V}}_{l+1}^{(i,:)}}{\left\|\widetilde{\bm{V}}_{l+1}^{(i,:)}\right\|} \min\left\{\left\|\widetilde{\bm{V}}_{l+1}^{(i,:)}\right\|,\sqrt{\frac{\mu c_sr}{n}}\right\}.
$$
\end{algorithmic}
\end{algorithm}

Algorithm~\ref{Resampling} begins with  partitioning the sampling set $\Omega$ into $L+1$ disjoint subsets. In each iteration, the new estimate is obtained via an application of FIHT on the new sampling set followed by the trimming procedure.  The use of a fresh sampling set in each iteration breaks the dependence  between the last estimate and the sampling set, while the trimming procedure ensures that the estimate remains  an $\mu_0$-incoherent matrix after each iteration. The following lemma provides an estimation of the approximation accuracy of the initial guess returned by Alg.~\ref{Resampling}.
\begin{lemma}\label{lem:resampling}
Assume $\mathcal{H}\bm{x}$ is  $\mu_0$-incoherent. Then with probability at least $1-(2L+1)n^{-2}$, the output of Alg.~\ref{Resampling} satisfies
$$
\|\widetilde{\bm{L}}_L-\mathcal{H}\bm{x}\|_F\leq \left(\frac56\right)^{L} \frac{\sigma_{\min}(\mathcal{H}\bm{x})}{256\kappa^2}
$$
provided $\widehat{m}\geq C \mu_0 c_s\kappa^6 r^2 \log(n)$ for some universal constant $C>0$.
\end{lemma}

We can obtain the following recovery guarantee for FIHT with $\BL_0$ being the output of Alg.~\ref{Resampling}.
\begin{theorem}[Guarantee II]\label{thm:resampling}
Assume $\mathcal{H}\bm{x}$ is $\mu_0$-incoherent. Let  $0<\varepsilon_0<\frac{1}{10}$ and $L=\left\lceil6\log\left(\frac{\sqrt{n}\log(n)}{16\varepsilon_0}\right)\right\rceil$. Define $\nu=10\vep_0<1$. Then with probability at least $1-\left(2L+3\right)n^{-2}$, the iterates generated by FIHT (Alg.~\ref{Alg-GD}) with $\bm{L}_0=\widetilde{\bm{L}}_L$ (the output of Alg.~\ref{Resampling}) satisfies  
$$
\|\bm{x}_{l}-\bm{x}\|\leq \nu^l\|\bm{L}_0-\mathcal{H}\bm{x}\|_F,
$$
provided
$$
m\geq C \mu_0 c_s\kappa^6 r^2 \log(n)\log\left(\frac{\sqrt{n}\log(n)}{16\varepsilon_0}\right)
$$ 
for some universal constant $C>0$.
\end{theorem}
\subsection{Spectrally Sparse Signal Reconstruction in Higher Dimensions}
Our results can be extended to higher dimensions based on the Hankel structures of multi-dimensional spectrally sparse signals. For concreteness, we discuss the three-dimensional setting but emphasize that the situation in general $d$ dimensions is similar. A $3$-dimensional array $\BX\in\C^{N_1\times N_2\times N_3}$ is spectrally sparse if
\begin{align*}
\BX\lb l_1,l_2,l_3\rb = \sum_{k=1}^rd_ky_k^{l_1}z_k^{l_2}w_k^{l_3},\quad\forall~(l_1,l_2,l_3)\in\lsb N_1\rsb\times\lsb N_2\rsb\times\lsb N_3\rsb
\end{align*}
with
\begin{align*}
y_k = \exp(2\pi\imath f_{1k}-\tau_{1k}),~z_k=\exp(2\pi\imath f_{2k}-\tau_{2k}), \mbox{ and }w_k=\exp(2\pi\imath f_{3k}-\tau_{3k})
\end{align*}
for some frequency triples $\bf_k=\lb f_{1k},f_{2k},f_{3k}\rb\in[0,1)^3$ and dampling factor triples $\bt_k=\lb \tau_{1k},\tau_{2k},\tau_{3k}\rb\in\mathbb{R}^3_{+}$.  Let $\Omega=\lcb (a_1,a_2,a_3)\in \lsb N_1\rsb\times\lsb N_2\rsb\times\lsb N_3\rsb\rcb$ be the set of indices for the known entries of $\BX$. The problem is to reconstruct $\BX$ from the partial known entries $\P_\Omega(\BX)$, which can be attempted by exploring the low rank Hankel structures as in one dimension.

The Hankel matrix corresponding to $\BX$ can be constructed recursively as follows 
\begin{align*}
\H\BX=\begin{bmatrix}
\H\BX(:,:,0)&\H\BX(:,:,1),&\cdots&\H\BX(:,:,N_3-n_3)\\
\H\BX(:,:,1)&\H\BX(:,:,2),&\cdots&\H\BX(:,:,N_3-n_3+1)\\
\vdots &\vdots &\ddots &\vdots\\
\H\BX(:,:,n_3-1)&\H\BX(:,:,n_3),&\cdots&\H\BX(:,:,N_3-1)
\end{bmatrix},
\end{align*}
where $\BX(:,:,l_3),~0\leq l_3< N_3$ is the $l_3$-th slice of $\BX$ and 
\begin{align*}
\H\BX(:,:,l_3)=\begin{bmatrix}
\H\BX(:,0,l_3)&\H\BX(:,1,l_3),&\cdots&\H\BX(:,N_2-n_2,l_3)\\
\H\BX(:,1,l_3)&\H\BX(:,2,l_3),&\cdots&\H\BX(:,N_2-n_2+1,l_3)\\
\vdots &\vdots &\ddots &\vdots\\
\H\BX(:,n_2-1,l_3)&\H\BX(:,n_2,l_3),&\cdots&\H\BX(:,N_2-1,l_3)
\end{bmatrix}.
\end{align*}
An explicit formula for $\H\BX$ is given by 
\begin{align*}
\lsb\H\BX\rsb_{ij}=\BX(l_1,l_2,l_3),
\end{align*}
where 
\begin{align*}
&i=i_1+i_2\cdot n_1+i_3\cdot n_1n_2,\\
&j = j_1+j_2\cdot (N_1-n_1+1)+j_3\cdot (N_1-n_1+1)(N_2-n_2+1),\\
&l_k=i_k+j_k, ~1\leq k \leq 3.
\end{align*}

There also exists a Vandermonde decomposition of $\H\BX$ of the form $\H\BX=\BE_L\BD\BE_R^T$, where the  $k$-th columns ($1\leq k\leq r$)
 of $\BE_L$ and $\BE_R$ are given by
{\begin{align*}
&\BE_L^{(:,k)} = \lcb y_k^{l_1}z_k^{l_2}w_k^{l_3},~(l_1,l_2,l_3)\in\lsb n_1\rsb\times\lsb n_2\rsb\times\lsb n_3\rsb\rcb,\\
&\BE_R^{(:,k)} = \lcb y_k^{l_1}z_k^{l_2}w_k^{l_3},~(l_1,l_2,l_3)\in\lsb N_1-n_1+1\rsb\times\lsb N_2-n_2+1\rsb\times\lsb N_3-n_3+1\rsb\rcb,
\end{align*}}and $\BD=\diag(d_1,\cdots,d_r)$ is a diagonal matrix. Therefore $\H\BX$ is still a rank $r$ matrix for high-dimensional arrays. To reconstruct $\BX$, we seek a three-dimensional array  that best fits the measurements and meanwhile corresponds to a rank $r$ Hankel matrix
\begin{align*}
\min_{\BZ\in\C^{N_1\times N_2\times N_3}}\la \BZ-\BX,\P_{\Omega}(\BZ-\BX)\ra\quad\mbox{subject to}\quad \rank(\mathcal{H}\bm{Z})=r.\numberthis\label{eq:3-array}
\end{align*}
The IHT (Alg.~\ref{IHT-GD}) and FIHT (Alg.~\ref{Alg-GD}) algorithms can be easily adapted for \eqref{eq:3-array}, with fast implementations for Hankel matrix-vector multiplications and the application of $\H^*$. Moreover, it can be established that $O(r^2\log^2(n))$ ($n=N_1N_2N_3$) number of measurements are sufficient for FIHT with resampling initialization to be able to reliably reconstruct spectrally sparse signals based on a similar incoherence notion for $\BE_L$ and $\BE_R$.
The details will be omitted for conciseness.
\subsection{Connections to Tight Frame Analysis Sparsity in Compressed Sensing}\label{sec:connection}
In its simplest form, compressed sensing \cite{donoho2006cs,CS} is about recovering a sparse vector from a number of linear measurements 
that is less than the length of the vector. Let $\bx\in\C^n$ be a vector with only $k$ nonzero entries and $\BA\in\C^{m\times n}$ be a measurement matrix from which we obtain $m\leq n$ measurements $\bb=\BA\bx$. Then compressed sensing attempts  to recover $\bx$ by finding a sparse vector that fits the measurements as well as possible 
\begin{align*}
\min_{\bz}\ln\BA\bz-\bb\rn^2\quad\mbox{subject to}\quad\ln\bz\rn_0=k,\numberthis\label{eq:cs_setup}
\end{align*}
where $\ln\bz\rn_0$ counts the number of nonzero entries in $\bz$.  
The simplest iterative hard thresholding algorithm for the compressed sensing problem is 
\begin{align*}
\bz_{l+1}=\T_k(\bz_l+\alpha\BA^*(\bb-\BA\bz_l)),\numberthis\label{alg:iht_cs}
\end{align*} 
where $\alpha$ is the line search stepsize and $\T_k$ denotes the hard thresholding operator which set all but the first  $k$ largest magnitude entries of a vector to zero. Theoretical recovery guarantees for \eqref{alg:iht_cs}  and its variants can be established  in items of the restricted isometry property of the measurement matrix $\BA$ \cite{bludav2009iht,blumensathdavies2010niht,foucart2011htp,CGIHT}.

However, in many real applications of interest,  the unknown vectors are not sparse, but instead they are sparse under some linear transforms.  For instance, though most of the natural images are not sparse,  they are usually sparse under  a class of wavelet or framelet transforms. For simplicity, we consider the tight frame analysis sparsity model which arises from a wide range of signal and image processing problems, see \cite{cai2008framelet,cai2010framelet,dong_shen_iciam} and references therein. Let $\BW\in\C^{d\times n}$ be a tight frame transform matrix which satisfies $\BW^T\BW=\bm{I}$. The tight frame analysis sparsity model assumes $\BW\bx$ is a sparse vector with only $k$ nonzero entries; that is $\ln\BW\bx\rn_0=k$ with $k\ll n$.
Then the compressed sensing problem under this assumption attempts to recover $\bx$ by seeking an analysis sparse vector which best fits the measurements
\begin{align*}
\min_{\bz}\ln\BA\bz-\bb\rn^2\quad\mbox{subject to}\quad\ln\BW\bz\rn_0=k.\numberthis\label{eq:cs_analysis_setup}
\end{align*}
An iterative hard thresholding algorithm can be developed for \eqref{eq:cs_analysis_setup} as follows by replacing $\T_k(\cdot)$ in \eqref{alg:iht_cs} with a composite hard thresholding operator $\BW^T\T_k\BW(\cdot)$,
\begin{align*}
\bz_{l+1}=\BW^T\T_k\BW(\bz_l+\alpha\BA^*(\bb-\BA\bz_l)).\numberthis\label{alg:iht_cs_analysis} 
\end{align*} 
 The {\em wavelet frame shrinkage operator} $\BW^T\T_k\BW(\cdot)$ has been widely used in signal and image processing  based on wavelet frame transforms, where $\T_k(\cdot)$ can also be the soft thresholding operator or other more complicated shrinkage operators; and \eqref{alg:iht_cs_analysis}  is typically referred to as the  {\em iterative wavelet frame shrinkage
algorithm} \cite{chan2003wavelet, dong_shen_iciam}.

There is a natural parallelization between the compressed sensing problem under the tight frame analysis sparsity model  \eqref{eq:cs_analysis_setup} and the spectrally sparse signal reconstruction problem via low rank Hankel matrix completion \eqref{eq:obj_fun}. In both problems,  the vectors to be reconstructed are not simple in the 
signal domain but simple in the transform domain. Therefore, in the iterative hard thresholding algorithms for these two problems the simple hard thresholding operators need to be replaced by the composite hard thresholding operators which first thresholding the vector in the transform domain and then synthesize the vector via the inverse transforms. A detailed comparison has been summarized in Tab.~\ref{tab:comp}.
\begin{table}[htp]
\centering
\caption{Parallelism between tight frame analysis sparsity in compressed sensing \eqref{eq:cs_analysis_setup} and low rank Hankel matrix completion in spectrally sparse signal reconstruction \eqref{eq:obj_fun}.}\label{tab:comp}
\vspace{0.2cm}
\makegapedcells
\setcellgapes{3pt}
\begin{tabular}{l|l}
\hline
\multicolumn{1}{c|}{\eqref{eq:cs_analysis_setup}}&\multicolumn{1}{c}{\eqref{eq:obj_fun}}\\
\hline
(a) $\bx$ is not sparse & (a) $\bx$ is not low rank\\
\hline 
(b) $\BW\bx$ is sparse, with $\BW^T\BW = \bm{I}$ & (b) $\H\bx$ is low rank, with $\H^\dag\H=\I$\\
\hline
\multirow{3}{*}{(c)  wavelet frame shrinkage $\BW^T\T_k\BW(\cdot)$ in \eqref{alg:iht_cs_analysis}}&
(c) low rank Hankel matrix thresholding \\&~~~~~~\llap{\textbullet} $\H^\dag\T_r\H(\cdot)$ in Alg.~\ref{IHT-GD}\\
&~~~~~~\llap{\textbullet} $\H^\dag\T_r\P_{\S_l}\H(\cdot)$ in Alg.~\ref{Alg-GD}\\
\hline
\end{tabular}
\end{table}

\section{Numerical Experiments}\label{sec:NumExp}
In this section, we conduct numerical experiments to evaluate the performance of IHT and FIHT. 
The experiments are  executed from Matlab 2014a on a MacBook Pro with a 2.7GHz dual-core Intel i5 CPU and 8 GB memory, and the algorithms are evaluated against successful  recovery rates, computational efficiency, robustness and capability of handling high-dimensional data.
We initialize IHT and FIHT using one step hard thresholding  computed via the PROPACK package \cite{PROPACK}
rather than the resampled FIHT (Alg.~\ref{Resampling}), as the former one has already shown very good performance and preliminary numerical results didn't present dramatic difference between those two initialization procedures for our simulations.
\subsection{Empirical Phase Transition}\label{sec:phase}
We investigate the recovery rates of IHT and FIHT in the framework of phase transition and compare them with EMaC \cite{Chi} and ANM \cite{Tang}. IHT and FIHT are terminated if the relative residual $\|\P_\Omega(\bm{x}_{rec})-\P_\Omega(\bm{x})\|_2/\|\P_\Omega(\bm{x})\|_2$ falls below $10^{-4}$ or $500$ number of iterations are reached. ANM and EMaC are implemented using CVX \cite{CVX} with default parameters. 
 The  spectrally sparse signals of length $n$ with $r$ frequency components are formed in the following way:  each frequency $f_k$ is uniformly sampled from $[0,1)$, and the argument of each complex coefficient $d_k$ is uniformly sampled from $[0,2\pi)$ while  the amplitude  is selected to be $1+10^{0.5c_k}$ with $c_k$ being uniformly distributed on $[0,1]$. Then $m$ entries of the test signals are sampled uniformly at random. 
For a given triple $(n,r,m)$, $50$ random tests are conducted.
We consider an algorithm to have successfully reconstructed a test signal if  $\|\bm{x}_{rec}-\bm{x}\|_2/\|\bm{x}\|\leq 10^{-3}$. The tests are conducted with $n=127$ and $p=m/n$ taking 18 equispaced values from 0.1 to 0.95. For a fixed pair of $(n,m)$, we start with $r=1$ and then increase it by one until it reaches a value such that the tested algorithm fails all the $50$ random tests.

\begin{figure}[!htb]
\centering
\begin{minipage}[b]{0.23 \textwidth}
\centering
	\subfigure{\includegraphics[width=1 \textwidth]{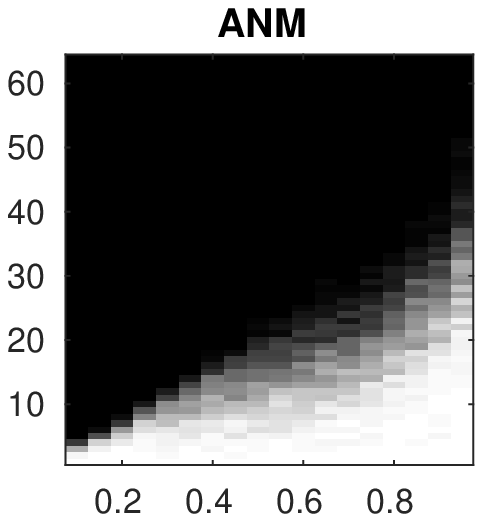}}
\end{minipage}
\begin{minipage}[b]{0.23 \textwidth}
\centering
	\subfigure{\includegraphics[width=1 \textwidth]{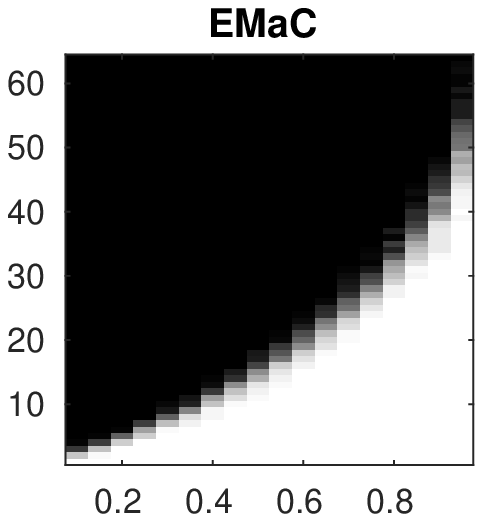}}
\end{minipage}
\begin{minipage}[b]{0.23 \textwidth}
\centering
	\subfigure{\includegraphics[width=1 \textwidth]{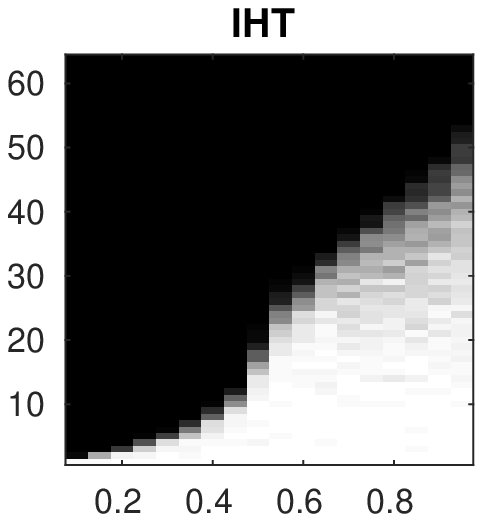}}
\end{minipage}
\begin{minipage}[b]{0.23 \textwidth}
\centering
	\subfigure{\includegraphics[width=1 \textwidth]{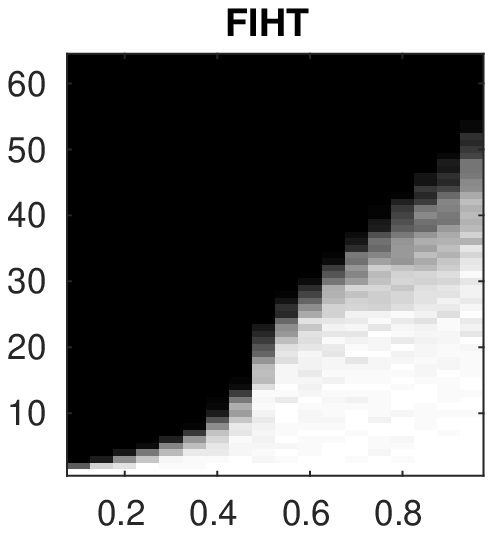}}
\end{minipage}
\vfill
\begin{minipage}[b]{0.23 \textwidth}
\centering
	\subfigure{\includegraphics[width=1 \textwidth]{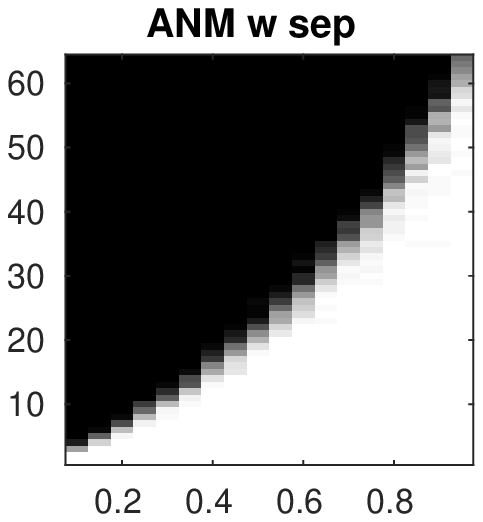}}
\end{minipage}
\begin{minipage}[b]{0.23 \textwidth}
\centering
	\subfigure{\includegraphics[width=1 \textwidth]{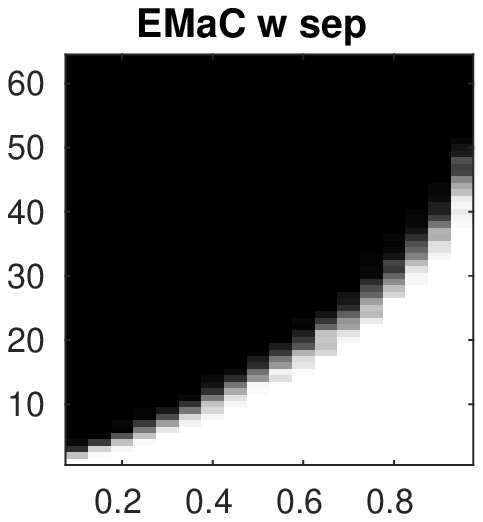}}
\end{minipage}
\begin{minipage}[b]{0.23 \textwidth}
\centering
	\subfigure{\includegraphics[width=1 \textwidth]{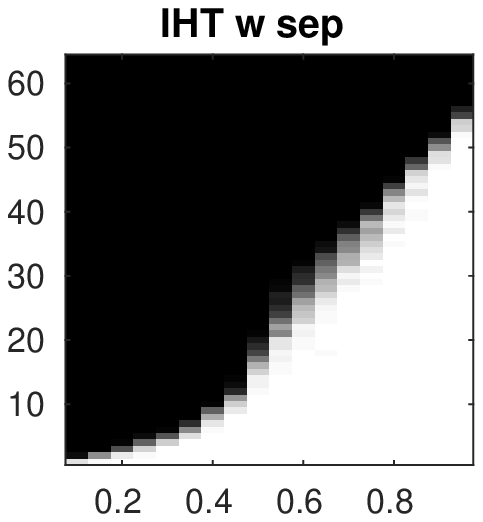}}
\end{minipage}
\begin{minipage}[b]{0.23 \textwidth}
\centering
	\subfigure{\includegraphics[width=1 \textwidth]{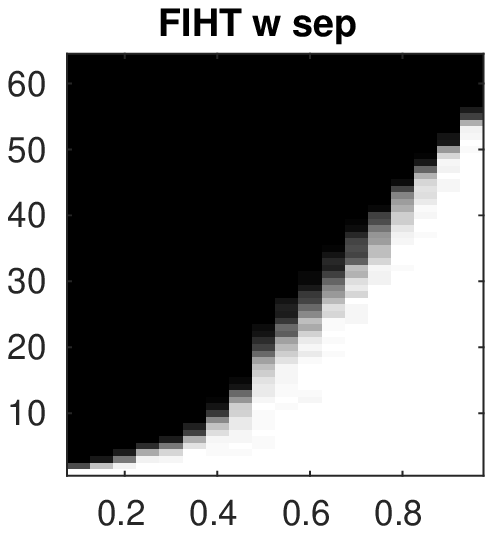}}
\end{minipage}
\caption{Phase transition comparisons:  $x$-axis is $p=m/n$ and $y$-axis is  $r$. Top: no restriction on frequencies of test signals; Bottom: wrap-around distances between frequencies is at least $1.5/n$.}
\label{fig:PT}
\end{figure}

The empirical phase transitions for the four tested algorithms ANM, EMaC, IHT and FIHT are presented in Fig.~\ref{fig:PT}, where  white color indicates that the algorithm can recover
all of the $50$ random test signals and on the other hand black color indicates the algorithm fails to recover each of
the randomly generated signals. The top four plots of the figure present the recovery phase transitions where no separation of the frequencies is imposed, while the bottom four plots presents the recovery phase transitions where the wrap-around distances between the randomly drawn frequencies are greater than $1.5/n$. First the figure shows that IHT and FIHT have similar empirical phase transitions for signals both with and without frequency separation. When the frequencies of test signals are separated, the phase transitions of IHT and FIHT are slightly lower than that of ANM, but higher than that of EMaC. The performance of ANM degrades severely  when the frequencies of test signals are not sufficiently separated, while IHT and FIHT can still achieve good performance.  The recovery phase transitions of EMaC seem to be irrelevant to the separation of frequencies.


\subsection{Computational Efficiency}\label{sec:nu_speed}

In this section, we compare IHT and FIHT  with PWGD on computational time. PWGD is an alternating projection algorithm which has been reported to be superior to ANM and EMaC in terms of computational efficiency \cite{PWGD}. In particular, we compare IHT and FIHT with an accelerated variant of PWGD based on  Nesterov's memory technique. In our experiments, PWGD is also initialized via one step hard thresholding and the parameters are tuned  as suggested in \cite{PWGD}.  The algorithms are tested with $n\in\{3999,7999\}$, $r\in\{15,30\}$ and $m\in\{800,1200\}$  and they are terminated whenever { $\|\bm{x}_{l+1}-\bm{x}_l\|_2/\|\bm{x}_l\|_2$ is less than $10^{-5}$}.
 For each triple $(n,r,m)$,  we run the algorithms on $10$ randomly generated problem instances where the signals are formed in the same way as in Sec.~\ref{sec:phase}. 
 The  average computational time  and average number of iterations for each tested algorithm are presented in Tab.~\ref{table:efficiency}.
The table shows that it takes almost the same number of iterations for IHT and FIHT to converge below the given tolerance, but FIHT requires about $1/3$ less computational time due to low per iteration computational complexity. Moreover, both IHT and FIHT are significantly faster than PWGD. 

\begin{table}[!htb]
\caption{{Average computational time (seconds) and average number of iterations of PWGD, IHT and FIHT over $10$ random problem instances per $(n,r,m)$ for $n\in\{3999,7999\}$, $r\in\{15,30\}$ and $m\in\{800,1200\}$.}}
\label{table:efficiency}
\begin{center}
{
\makegapedcells
\setcellgapes{3pt}
\small
\begin{tabular}{ccccccccccccc}
\hline
\multicolumn{1}{|c|}{$r$} & \multicolumn{6}{c|}{15} & \multicolumn{6}{c|}{30} \\
\hline
\multicolumn{1}{|c|}{$m$} & \multicolumn{3}{c|}{800} & \multicolumn{3}{c|}{1200} & \multicolumn{3}{c|}{800} & \multicolumn{3}{c|}{1200} \\
\hline
\multicolumn{1}{|c|}{} & \multicolumn{1}{c}{rel.err} & \multicolumn{1}{c}{iter} & \multicolumn{1}{c|}{time} & \multicolumn{1}{c}{rel.err} & \multicolumn{1}{c}{iter} & \multicolumn{1}{c|}{time} & \multicolumn{1}{c}{rel.err} & \multicolumn{1}{c}{iter} & \multicolumn{1}{c|}{time} & \multicolumn{1}{c}{rel.err} & \multicolumn{1}{c}{iter} & \multicolumn{1}{c|}{time} \\
\hline
\multicolumn{1}{|c|}{} & \multicolumn{12}{c|}{$n$=3999} \\
\hline
\multicolumn{1}{|c|}{PWGD} & 9e-6 & 55 & \multicolumn{1}{c|}{6.28} & 7.4e-6 & 35 & \multicolumn{1}{c|}{3.92} & 9.4e-6 & 71 & \multicolumn{1}{c|}{16.88} & 9e-6 & 42 & \multicolumn{1}{c|}{9.99} \\ 
\hline
\multicolumn{1}{|c|}{IHT} & 7.2e-6 & 12 & \multicolumn{1}{c|}{1.18} & 4.9e-6 & 9 & \multicolumn{1}{c|}{0.89} & 7.8e-6 & 19 & \multicolumn{1}{c|}{3.71} & 6.6e-6 & 13 & \multicolumn{1}{c|}{2.54} \\
\hline
\multicolumn{1}{|c|}{FIHT} & 6.1e-6 & 12 & \multicolumn{1}{c|}{0.70} & 6.2e-6 & 9 & \multicolumn{1}{c|}{0.53} & 6.8e-6 & 19 & \multicolumn{1}{c|}{1.98} & 6.9e-6 & 12 & \multicolumn{1}{c|}{1.41} \\
\hline
\multicolumn{1}{|c|}{} & \multicolumn{12}{c|}{$n$=7999} \\
\hline
\multicolumn{1}{|c|}{PWGD} & 9.5e-6 & 98 & \multicolumn{1}{c|}{27.69} & 8.7e-6 & 61 & \multicolumn{1}{c|}{17.49} & 9.6e-6 & 150 & \multicolumn{1}{c|}{97.36} & 9.3e-6 & 75 & \multicolumn{1}{c|}{48.95} \\ 
\hline
\multicolumn{1}{|c|}{IHT}  & 6.6e-6 & 13 & \multicolumn{1}{c|}{3.49} & 6.2e-6 & 10 & \multicolumn{1}{c|}{2.81} & 8.3e-6 & 24 & \multicolumn{1}{c|}{14.03} & 7.3e-6 & 15 & \multicolumn{1}{c|}{8.86} \\
\hline
\multicolumn{1}{|c|}{FIHT}  & 6.9e-6 & 12 & \multicolumn{1}{c|}{2.31} & 6.3e-6 & 10 & \multicolumn{1}{c|}{1.94} & 8e-6 & 23 & \multicolumn{1}{c|}{8.34} & 6.9e-6 & 14 & \multicolumn{1}{c|}{5.36} \\
\hline
\end{tabular}}
\end{center}
\end{table}


\subsection{Robustness to Additive Noise}\label{sec:robust}
We demonstrate the performance of IHT and FIHT under additive noise by conducting tests with the measurements 
 corrupted by the vector
$$
e=\sigma \cdot \|\mathcal{P}_{\Omega}(\bm{x})\|_2 \cdot \frac{\bm{w}}{\|\bm{w}\|_2},
$$
where $\bx$ is the random signal to be reconstructed, the entries of $\bm{w}$ are i.i.d. standard Gaussian random variables and $\sigma$ is referred to as the noise level. 

Tests are conducted with $9$ different values of $\sigma$ from $10^{-4}$ to 1, corresponding to $9$ equispaced signal-to-noise ratios (SNR) from 80 to 0 dB. For each $\sigma$, 10 random problem instances are tested  and the algorithms are terminated when $\|\bm{x}_{l+1}-\bm{x}_l\|_2/\|\bm{x}_l\|_2< 10^{-5}$.  The average relative reconstruction error in dB plotted against the SNR is presented in Fig.~\ref{fig:robustness} for IHT and FIHT. The
figure clearly shows the desirable linear scaling between the noise levels and the relative reconstruction errors for both IHT and FIHT.  It  can be further observed that the reconstruction error decreases as the number of measurements increases for both algorithms.



\begin{figure}[!htb]\label{fig:robustness}
\centering
\subfigure[]{
\begin{minipage}[b]{0.4 \textwidth}
\centering
	\includegraphics[width=1 \textwidth]{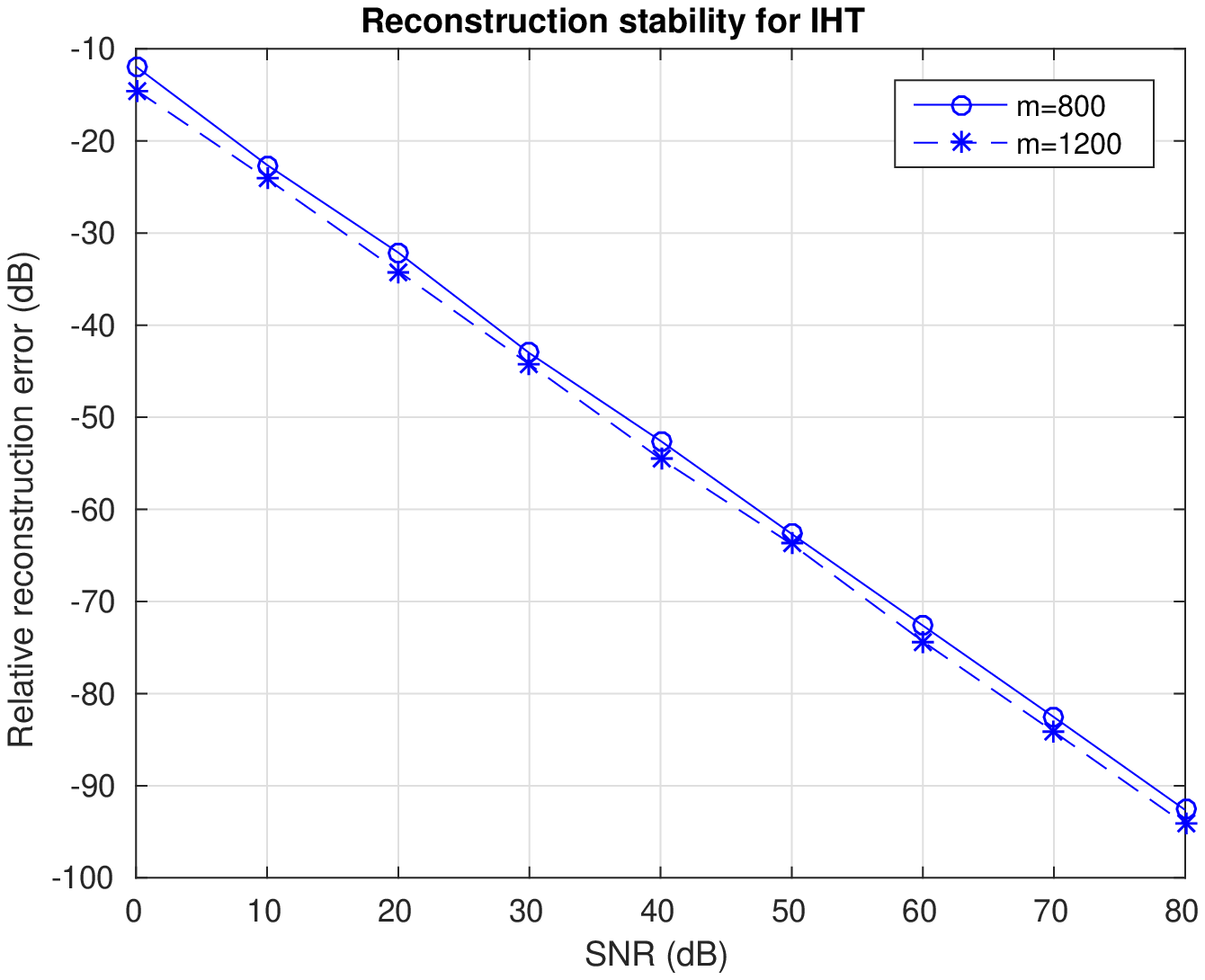}
\end{minipage}}
\subfigure[]{
\begin{minipage}[b]{0.4 \textwidth}
\centering
	\includegraphics[width=1 \textwidth]{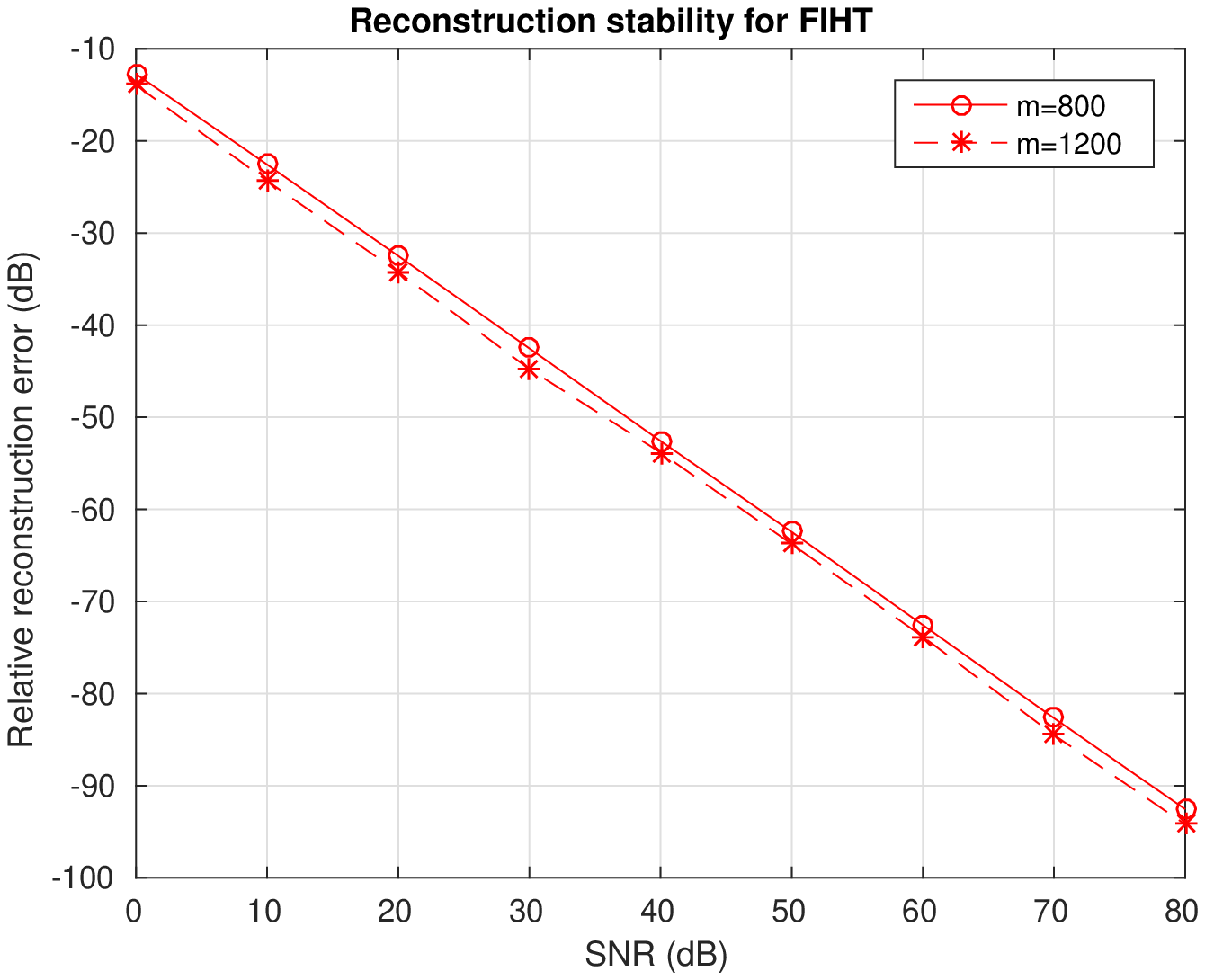}
\end{minipage}}
\caption{Performances of (a) IHT and (b) FIHT under different SNR.}
\end{figure}


 
\subsection{A 3D Example}


To explore the capability of FIHT on handling large data, we conduct tests on a $3$D damped signal with $n=N_1\times N_2\times N_3=31\times 31\times 511=491071$, $r=10$ and $m=19642$ (about $4\%$ of $n$). The signal is constructed to simulate real data from Nuclear Magnetic Resonance (NMR) spectroscopy. In this experiment, FIHT is terminated when 
$\|\bm{x}_{l+1}-\bm{x}_l\|_2/\|\bm{x}_l\|_2< 10^{-5}$. It takes FIHT $\bm{39}$ iterations and $\bm{1554}$ seconds to converge below the tolerance with the relative reconstruction error being $3.95\times 10^{-6}$.

To visualize the reconstruction result, we randomly pick a slice of the 3D signal  and plot the amplitudes of sampled and reconstructed entries on this slice in  Fig.~\ref{fig:3D-rec}. The differences between each entry of  the original and reconstructed signals on the same slice is plotted in Fig.~\ref{fig:3D-diff}, which shows that the reconstruction is very accurate.  Furthermore, the plots in Fig.~\ref{fig:3D-spectrum} compare the projection spectra of the original signal and the reconstructed one, which is obtained by first taking {the Fourier transform of the $3$D signal and then sum the spectrum} along the third dimension.

\begin{figure}[htb]
\centering
	\includegraphics[width=0.4 \textwidth]{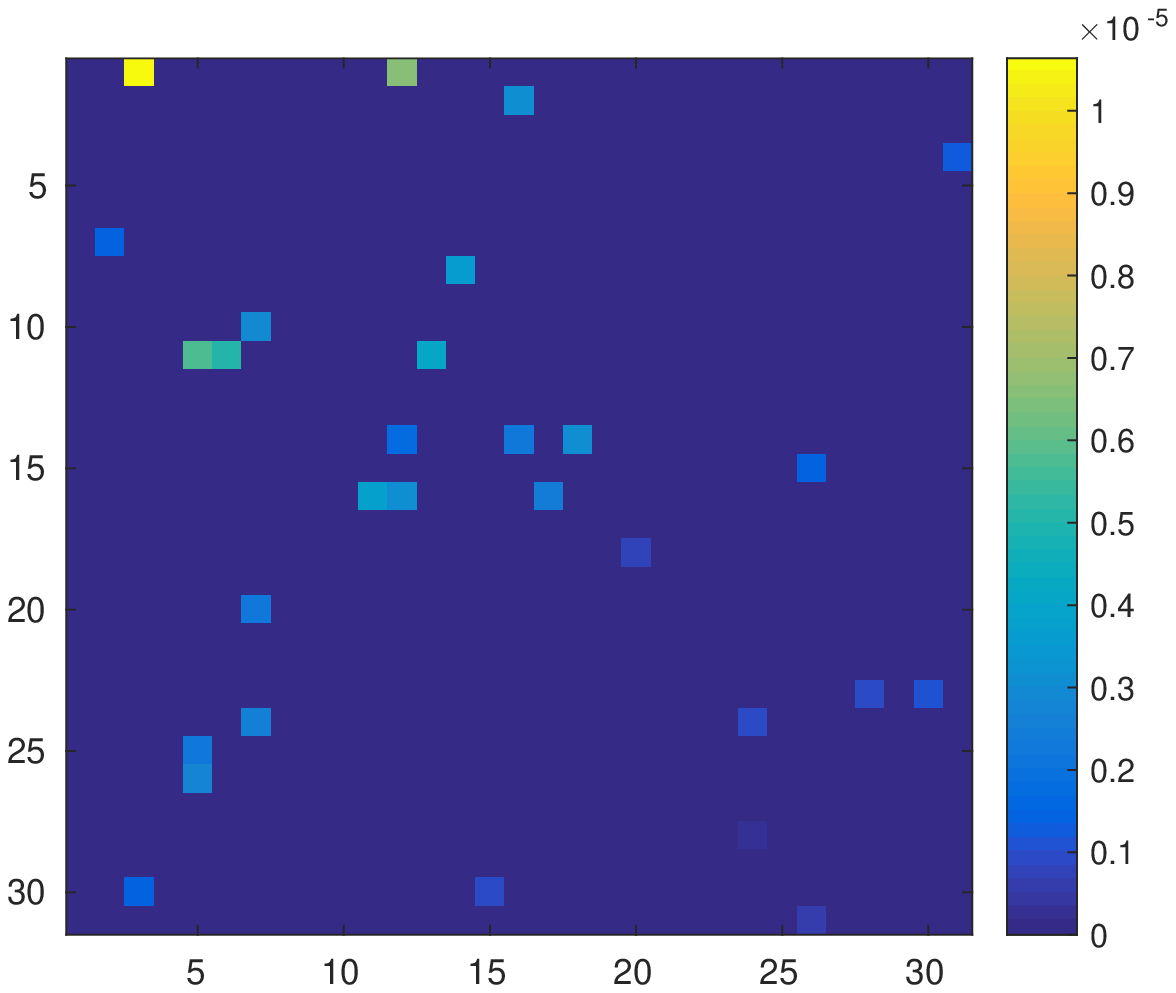}
	\includegraphics[width=0.4 \textwidth]{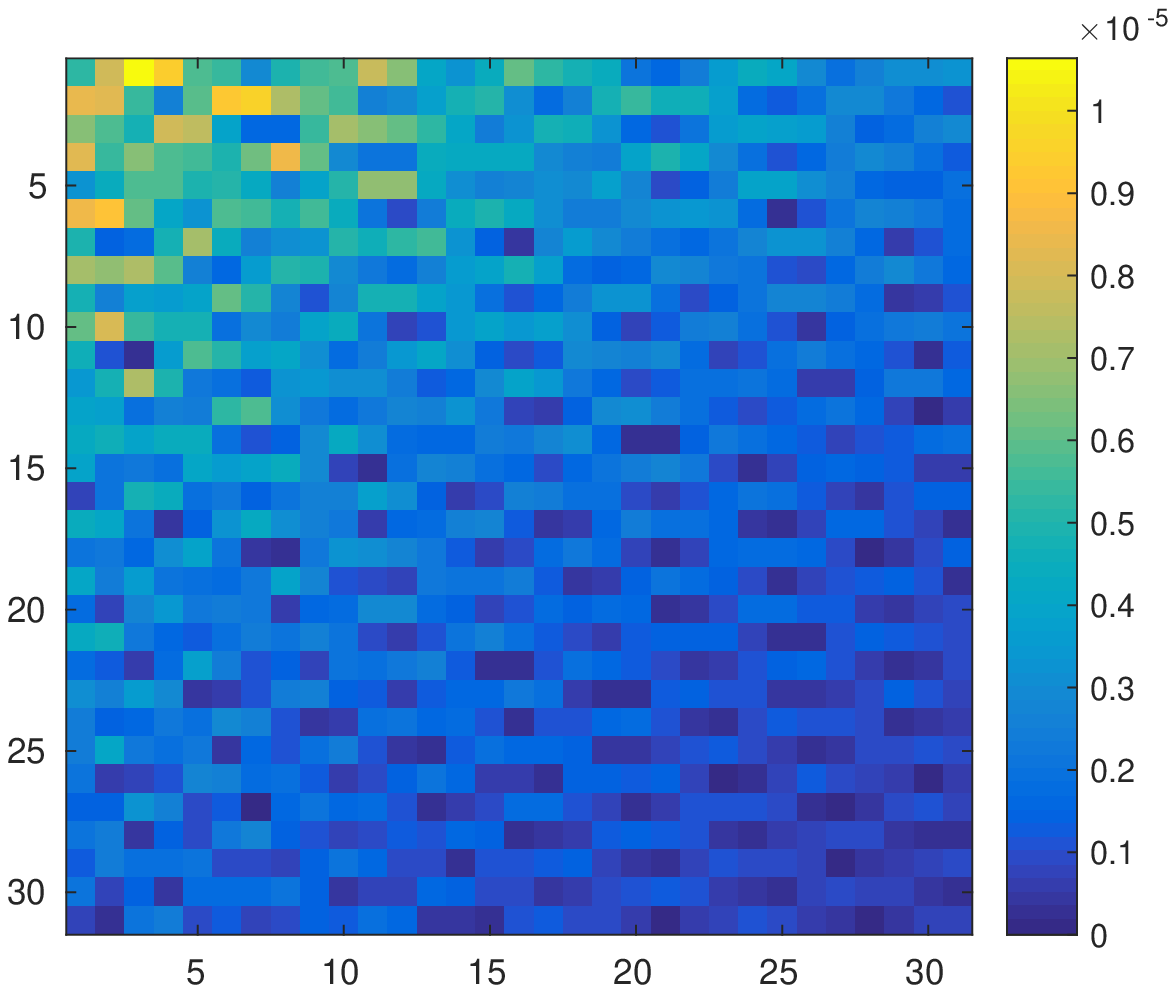}
\caption{Samples (Left) on the slice with $N_3=491$ and its reconstruction (Right).}
\label{fig:3D-rec}
\end{figure}
 
\begin{figure}[htb]
\centering
	\includegraphics[width=0.4 \textwidth]{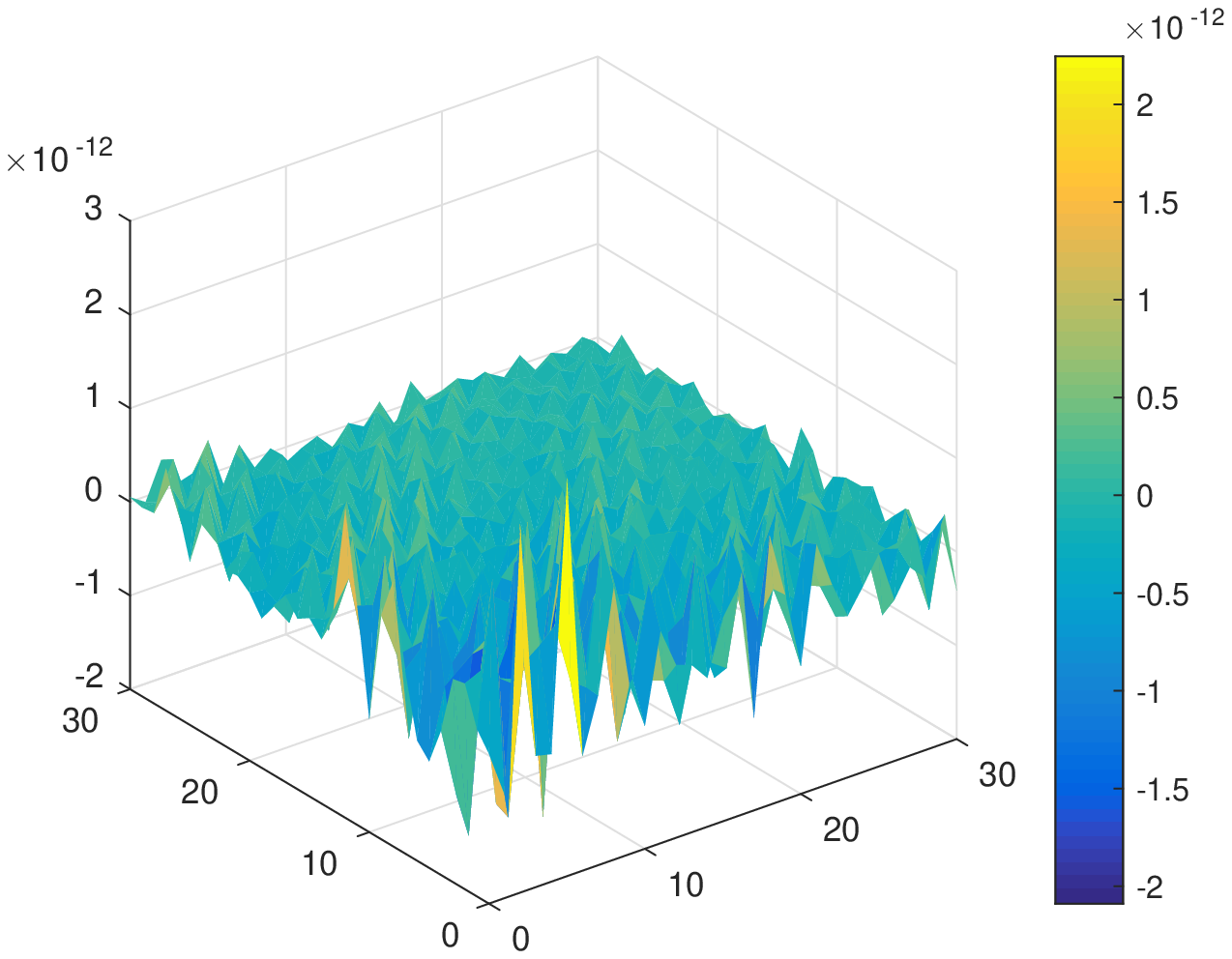}
	\includegraphics[width=0.4 \textwidth]{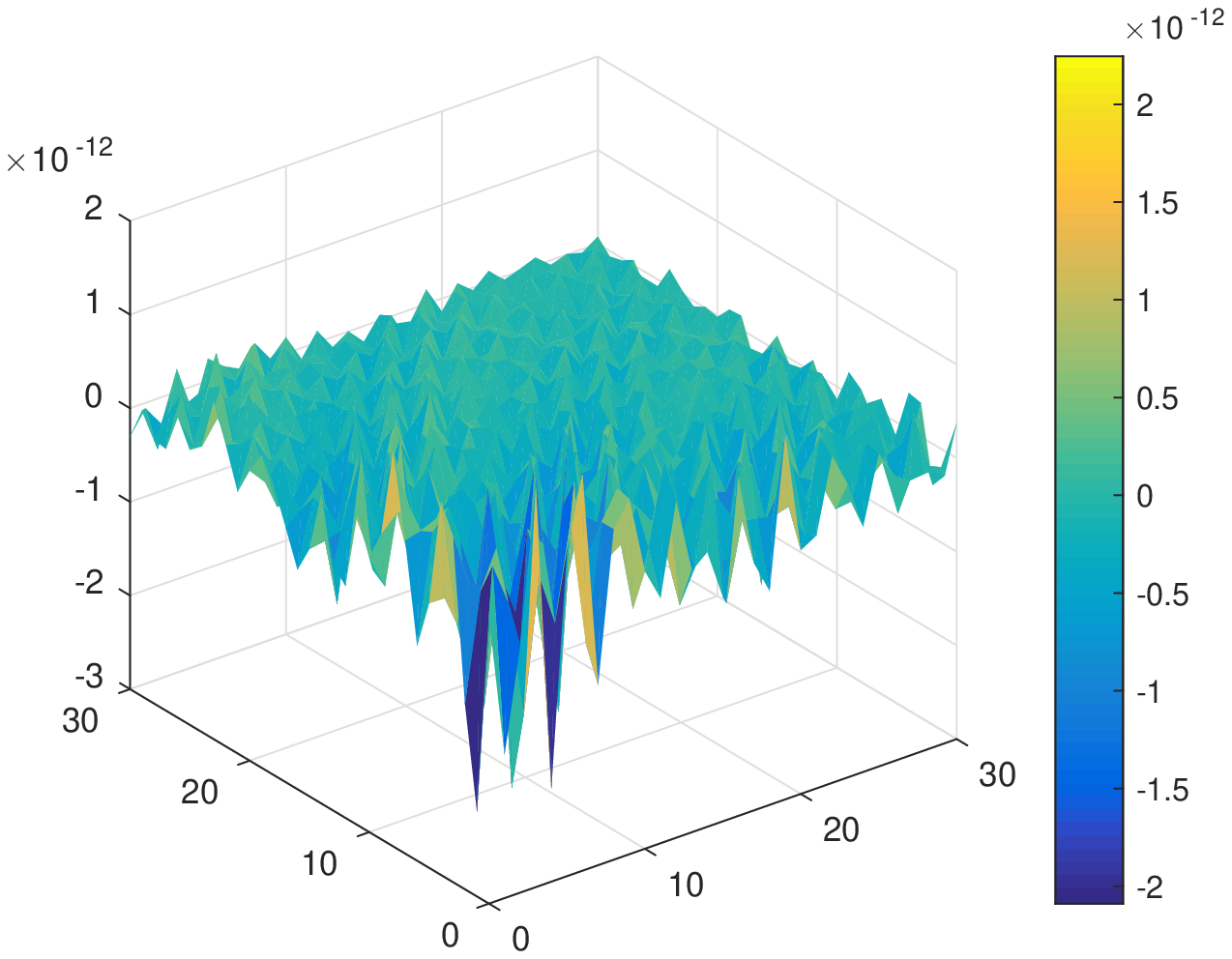}
\caption{Real (Left) and imaginary (Right) parts of  reconstruction errors  for each entry on the slice with $N_3=491$.}
\label{fig:3D-diff}
\end{figure}

\begin{figure}[htb]
\centering
	\includegraphics[width=0.4 \textwidth]{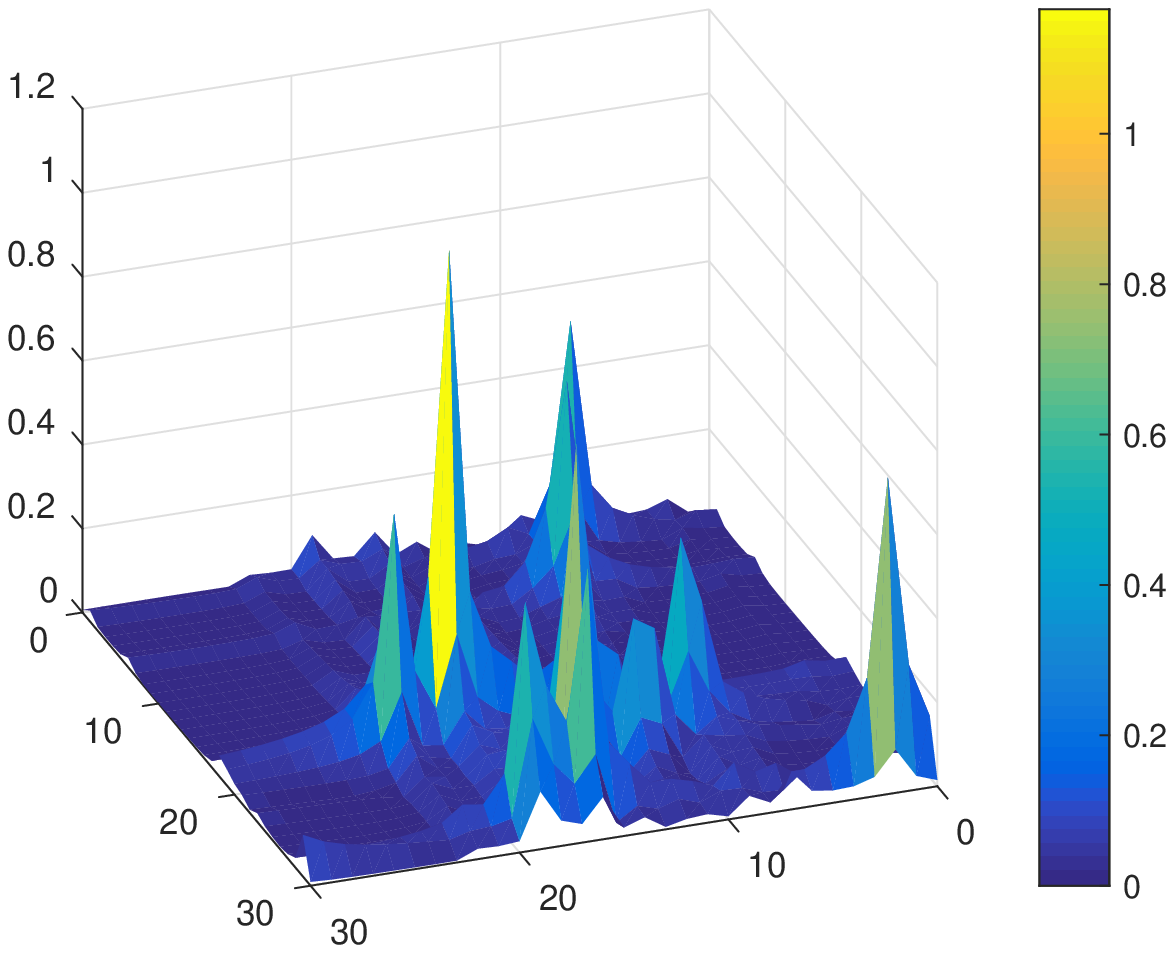}
	\includegraphics[width=0.4 \textwidth]{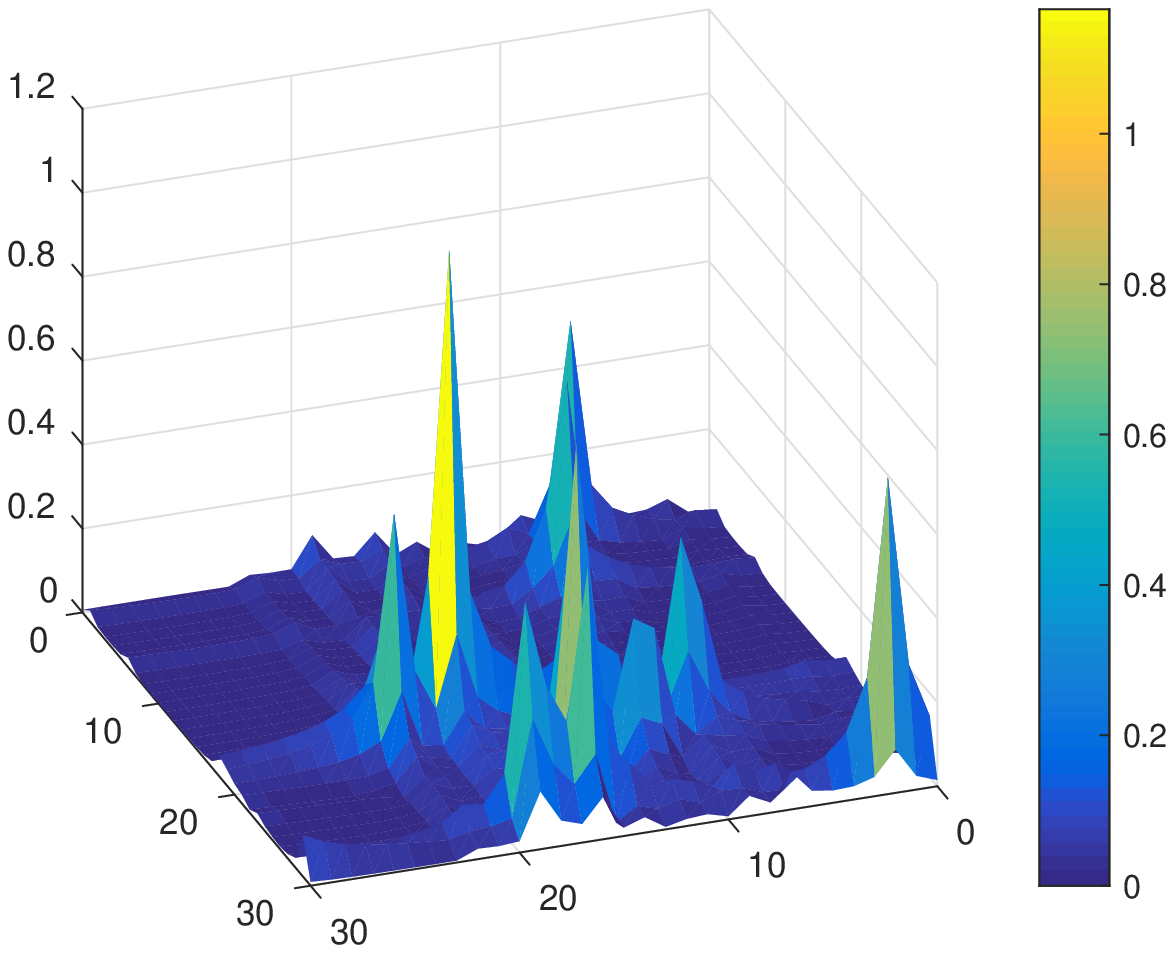}
\caption{Projection spectra of the original signal (Left) and its reconstruction (Right).}
\label{fig:3D-spectrum}
\end{figure}

\section{Proofs}\label{sec:proofs}
This section presents the proofs for the theoretical results in Sec.~\ref{sec:theory}. 
We first introduce several new variables and notation. Recall that $\H$ is a Hankel operator which maps a vector to a Hankel matrix and $\H^*$ is the adjoint of $\H$. Moreover, $\D^2=\H^*\H=\diag(w_0,\cdots,w_{n-1})$ is a diagonal operator which multiply the $a$-th entry of  a vector by the number of elements in the $a$-th  anti-diagonal of the corresponding Hankel matrix. Recall that $\lcb\BH_a\rcb_{a=0}^{n-1}\subset\C^{n_1\times n_2}$ forms an orthonormal basis for all the $n_1\times n_2$ Hankel matrices with $n_1+n_2=n+1$.

Define $\G=\H\D^{-1}$.  Then the adjoint of $\G$ is given by $\G^*=\D^{-1}\H^*$. It can be easily verified that $\G$ and $\G^*$ have the following properties: 
\begin{itemize}
\item $\G^*\G=\I$, $\ln\G\rn\leq 1$, and $\ln\G^*\rn\leq 1$;
\item $\G\bz=\sum_{a=0}^{n-1}z_a\BH_a,~\forall\bz\in\C^n$;
\item$\G^*\BZ = \lcb\la\BZ,\BH_a\ra\rcb_{a=0}^{n-1},~\forall\BZ\in\C^{n_1\times n_2}$.
\end{itemize}
Notice that the iteration of FIHT (Alg.~\ref{Alg-GD}) can be written in a compact form 
\begin{align*}
\bm{x}_{l+1}=\mathcal{H}^{\dag}\mathcal{T}_r\mathcal{P}_{\mathcal{S}_l}\mathcal{H}(\bm{x}_l+p^{-1}\mathcal{P}_{\Omega}(\bm{x}-\bm{x}_l)).\numberthis\label{eq:x_compact}
\end{align*}
So if we define $\by=\D\bx$ and $\by_l=\D\bx_l$, the following iteration can be established for $\by_l$
\begin{align*}
\bm{y}_{l+1}=\mathcal{G}^{*}\mathcal{T}_{r}\mathcal{P}_{\mathcal{S}_l}\mathcal{G}(\bm{y}_l+p^{-1}\mathcal{P}_{\Omega}(\bm{y}-\bm{y}_l))\numberthis\label{eq:y_compact}
\end{align*}
since $\P_\Omega$ and $\D^{-1}$ commute with each other. 
{\em For ease of exposition, we will prove the lemmas and theorems in Sec.~\ref{sec:theory} in terms of $\by_l$ and $\by$ but note that the results in terms of $\bx_l$ and $\bx$ follow immediately  since $\H\bx=\G\by$ and
\begin{align*}
\ln\bx_l-\bx\rn=\ln\D^{-1}(\by_l-\by)\rn\leq\ln\by_l-\by\rn.\numberthis\label{eq:err_x_y}
\end{align*}}

The following supplementary results from the literature but using our notation will be used repeatedly in the proofs of the main results. 
\begin{lemma}[{\cite[Proposition 3.3]{P_Omega}}]\label{lem:sampling}
Under the sampling with replacement model, the maximum number of repetitions of any entry in $\Omega$ is less than $8\log(n)$ with probability at least $1-n^{-2}$ provided $n\geq 9$.
\end{lemma}
\begin{lemma}[{\cite[Lemma 3]{Chi}}]\label{lem:RIPPS}
Let $\BU\in\C^{n_1\times r}$ and $\BV\in\C^{n_2\times r}$ be two orthogonal matrices which satisfy
\begin{align*}
\ln\P_{\BU}(\BH_a)\rn_F^2\leq\frac{\mu c_sr}{n}\quad&\mbox{and}\quad \ln\P_{\BV}(\BH_a)\rn_F^2\leq \frac{\mu c_sr}{n}.
\end{align*}
 Let $\S$ be the subspace defined in \eqref{eq:subspace}. Then
\begin{equation}
\|\mathcal{P}_{\mathcal{S}}\mathcal{G}\mathcal{G}^{*}\mathcal{P}_{\mathcal{S}}-p^{-1}\mathcal{P}_{\mathcal{S}}\mathcal{G}\mathcal{P}_{\Omega}\mathcal{G}^{*}\mathcal{P}_{\mathcal{S}}\|
\leq \sqrt{\frac{32\mu c_sr\log(n)}{m}}
\end{equation}
holds with probability at least $1-n^{-2}$ provided that 
$$
m\geq32\mu c_sr\log(n).
$$ 
\end{lemma}
\begin{lemma}[{\cite[Lemma 4.1]{Completion}}]\label{lem:Riemannian}
Let $\BL_l=\BU_l\BS_l\BV_l^*$ be another rank $r$ matrix and $\S_l$ be the tangent space of the rank $r$ matrix manifold at $\BL_l$ as defined in \eqref{eq:subspace}. Then
$$
\|(\mathcal{I}-\mathcal{P}_{\mathcal{S}_l})(\bm{L}_l-\mathcal{G}\bm{y})\|_F\leq\frac{\|\bm{L}_l-\mathcal{G}\bm{y}\|_F^2}{\sigma_{\min}(\mathcal{G}\bm{y})}, \quad
\|\mathcal{P}_{\mathcal{S}_l}-\mathcal{P}_{\mathcal{S}}\|\leq \frac{2\|\bm{L}_l-\mathcal{G}\bm{y}\|_F}{\sigma_{\min}(\mathcal{G}\bm{y})}.
$$
\end{lemma}
\begin{lemma}[{\cite[Theorem 1.6]{Tropp}}]\label{lem:bernstein}
Consider a finite sequence $\lcb\BZ_k\rcb$ of independent, random matrices with dimensions $d_1\times d_2$. Assume that each random matrix satisfies 
\begin{align*}
\E\lb\BZ_k\rb=0\quad\mbox{and}\quad\ln\BZ_k\rn\leq R\quad\mbox{almost surely}.
\end{align*}
Define 
\begin{align*}
\sigma^2:=\max\lcb \ln\sum_k\E\lb\BZ_k\BZ_k^*\rb\rn, \ln\sum_k\E\lb\BZ_k^*\BZ_k\rb\rn\rcb.
\end{align*}
Then for all $t\ge 0$,
\begin{align*}
\mathbb{P}\lcb\ln\sum_k\BZ_k\rn\geq t\rcb\leq (d_1+d_2)\exp\lb\frac{-t^2/2}{\sigma^2+Rt/3}\rb.
\end{align*}
\end{lemma}
\subsection{Local Convergence}
We begin with a deterministic convergence result which  characterizes the ``basin of attraction'' for FIHT. If the initial guess is located in this attraction region, FIHT will converge linearly to the underlying true solution.
\begin{theorem}\label{thm:local}
Assume $0<\vep_0<\frac{1}{10}$ and the following conditions
\begin{align*}
&\|\P_{\Omega}\|\leq 8\log(n),\numberthis\label{c1}\\
&\|\mathcal{P}_{\mathcal{S}}\mathcal{G}\mathcal{G}^{*}\mathcal{P}_{\mathcal{S}}-p^{-1}\mathcal{P}\mathcal{G}\mathcal{P}_{\Omega}\mathcal{G}^{*}\mathcal{P}_{\mathcal{S}}\|
\leq\vep_0, \numberthis\label{c2}\\
&\frac{\|\bm{L}_0-\mathcal{G}\bm{y}\|_F}{\sigma_{\min}(\mathcal{G}\bm{y})}\leq
\frac{p^{1/2}\vep_0}{16\log(n)(1+\vep_0)} \numberthis\label{c3}
\end{align*}
are satisfied. Then the iterate $\by_l$ in \eqref{eq:y_compact} satisfies
$
\|\bm{y}_l-\bm{y}\|\leq \nu^{l} \|\bm{L}_0-\G\by\|_F
$
with $\nu=10\vep_0<1$. 
\end{theorem}
The proof of Thm.~\ref{thm:local} makes use of the restricted isometry property of $\P_\Omega(\cdot)$ on $\S_l$ when 
$\BL_l$ is in a small neighborhood  of $\G\by$.
\begin{lemma}\label{lem:RIPPL}
Suppose \eqref{c1}, \eqref{c2} hold and 
\begin{align*}
\frac{\|\bm{L}_l-\mathcal{G}\bm{y}\|_F}{\sigma_{\min}(\mathcal{G}\bm{y})}\leq
\frac{p^{1/2}\vep_0}{16\log(n)(1+\vep_0)}. \numberthis\label{c4}
\end{align*}
 Then we have 
\begin{equation}\label{eq:POGPS}
\|\mathcal{P}_{\Omega}\mathcal{G}^{*}\mathcal{P}_{\mathcal{S}_l}\|
\leq8\log(n) (1+\vep_0)p^{1/2}
\end{equation}
and 
\begin{equation}\label{eq:RIPPL}
\|\mathcal{P}_{\mathcal{S}_{l}}\mathcal{G}\mathcal{G}^*\mathcal{P}_{\mathcal{S}_{l}}-p^{-1}\mathcal{P}_{\mathcal{S}_{l}}\mathcal{G}\mathcal{P}_{\Omega}\mathcal{G}^*\mathcal{P}_{\mathcal{S}_{l}}\|
\leq 4\vep_0.
\end{equation}
\end{lemma}
\begin{proof}
Since $\|\mathcal{P}_{\mathcal{S}}\mathcal{G}\mathcal{P}_{\Omega}\|=\|(\mathcal{P}_{\mathcal{S}}\mathcal{G}\mathcal{P}_{\Omega})^{*}\|=\|\mathcal{P}_{\Omega}\mathcal{G}^*\mathcal{P}_{\mathcal{S}}\|$, for any $\bm{Z}\in\mathbb{C}^{n_1\times n_2}$,
\begin{equation*}
\begin{split}
\|\mathcal{P}_{\Omega}\mathcal{G}^{*}\mathcal{P}_{\mathcal{S}}(\BZ)\|^2
&=\langle\mathcal{P}_{\Omega}\mathcal{G}^{*}\mathcal{P}_{\mathcal{S}}(\BZ),\mathcal{P}_{\Omega}\mathcal{G}^{*}\mathcal{P}_{\mathcal{S}}(\BZ)\rangle \cr
&\leq 8 \log(n) \langle\mathcal{G}^{*}\mathcal{P}_{\mathcal{S}}(\BZ),\mathcal{P}_{\Omega}\mathcal{G}^{*}\mathcal{P}_{\mathcal{S}}(\BZ)\rangle \cr
&=8 \log(n) \langle \BZ,\mathcal{P}_{\mathcal{S}}\mathcal{G}\mathcal{P}_{\Omega}\mathcal{G}^{*}\mathcal{P}_{\mathcal{S}}(\BZ)\rangle \cr
&\leq 8 \log(n) (1+\vep_0)p \|\BZ\|_F^2
\end{split}
\end{equation*}
where the first inequality follows from \eqref{c1} and the second inequality follows from \eqref{c2}.
So it follows that $\|\mathcal{P}_{\mathcal{S}}\mathcal{G}\mathcal{P}_{\Omega}\|=\|\mathcal{P}_{\Omega}\mathcal{G}^*\mathcal{P}_{\mathcal{S}}\|\leq \sqrt{8 \log(n) (1+\vep_0)p}$ and
\begin{equation*}
\begin{split}
\|\mathcal{P}_{\Omega}\mathcal{G}^{*}\mathcal{P}_{\mathcal{S}_l}\|
&\leq \|\mathcal{P}_{\Omega}\mathcal{G}^{*}(\mathcal{P}_{\mathcal{S}_l}-\mathcal{P}_{\mathcal{S}})\|+\|\mathcal{P}_{\Omega}\mathcal{G}^{*}\mathcal{P}_{\mathcal{S}}\|\cr
&\leq 8 \log(n) \frac{2\|\bm{L}_l-\mathcal{G}\bm{y}\|_F}{\sigma_{\min}(\mathcal{G}\bm{y})}+\|\mathcal{P}_{\Omega}\mathcal{G}^{*}\mathcal{P}_{\mathcal{S}}\|\cr
&\leq 8 \log(n) \frac{p^{1/2}\vep_0}{8\log(n)(1+\vep_0)}+\sqrt{8 \log(n) (1+\vep_0)p} \cr
&\leq 8 \log(n) (1+\vep_0)p^{1/2},
\end{split}
\end{equation*}
where the second inequality follows from \eqref{c1} and Lem.~\ref{lem:Riemannian},  the third inequality follows from \eqref{c4}.

Finally,
\begin{align*}
&\|\mathcal{P}_{\mathcal{S}_{l}}\mathcal{G}\mathcal{G}^*\mathcal{P}_{\mathcal{S}_{l}}-p^{-1}\mathcal{P}_{\mathcal{S}_{l}}\mathcal{G}\mathcal{P}_{\Omega}\mathcal{G}^*\mathcal{P}_{\mathcal{S}_{l}}\|\\
&\leq
\|\mathcal{P}_{\mathcal{S}}\mathcal{G}\mathcal{G}^*\mathcal{P}_{\mathcal{S}}-p^{-1}\mathcal{P}_{\mathcal{S}}\mathcal{G}\mathcal{P}_{\Omega}\mathcal{G}^*\mathcal{P}_{\mathcal{S}}\|
+\|(\mathcal{P}_{\mathcal{S}}-\mathcal{P}_{\mathcal{S}_{l}})\mathcal{G}\mathcal{G}^*\mathcal{P}_{\mathcal{S}_{l}}\|+\|\mathcal{P}_{\mathcal{S}}\mathcal{G}\mathcal{G}^*(\mathcal{P}_{\mathcal{S}}-\mathcal{P}_{\mathcal{S}_{l}})\|\\
&\quad+\|p^{-1}(\mathcal{P}_{\mathcal{S}}-\mathcal{P}_{\mathcal{S}_{l}})\mathcal{G}\mathcal{P}_{\Omega}\mathcal{G}^*\mathcal{P}_{\mathcal{S}_{l}}\|
+\|p^{-1}\mathcal{P}_{\mathcal{S}}\mathcal{G}\mathcal{P}_{\Omega}\mathcal{G}^*(\mathcal{P}_{\mathcal{S}}-\mathcal{P}_{\mathcal{S}_{l}})\|\\
&\leq\vep_0+ \frac{4\|\bm{L}_l-\mathcal{G}\bm{y}\|}{\sigma_{\min}(\mathcal{G}\bm{y})}+p^{-1}\cdot \frac{2\|\bm{L}_l-\mathcal{G}\bm{y}\|}{\sigma_{\min}(\mathcal{G}\bm{y})}\cdot (\|\mathcal{P}_{\Omega}\mathcal{G}^{*}\mathcal{P}_{\mathcal{S}_l}\|+\|\mathcal{P}_{\mathcal{S}}\mathcal{G}\mathcal{P}_{\Omega}\|)\\
&\leq 4\vep_0,
\end{align*}
which completes the proof of \eqref{eq:RIPPL}.
\end{proof}
\begin{proof}[Proof of Theorem~\ref{thm:local}]
First note that $\BL_{l+1}=\T_r(\BW_l)$, where
\begin{align*}
\BW_l&=\mathcal{P}_{\mathcal{S}_l}\mathcal{H}(\bm{x}_l+p^{-1}\mathcal{P}_{\Omega}(\bm{x}-\bm{x}_l))\\
&=\mathcal{P}_{\mathcal{S}_l}\mathcal{G}(\bm{y}_l+p^{-1}\mathcal{P}_{\Omega}(\bm{y}-\bm{y}_l)).
\end{align*}
So we have 
\begin{align*}
\|\bm{L}_{l+1}-\mathcal{G}\bm{y}\|_F
&\leq \|\bm{W}_{l}-\bm{L}_{l+1}\|_F+\|\bm{W}_{l}-\mathcal{G}\bm{y}\|_F\leq2\|\bm{W}_{l}-\mathcal{G}\bm{y}\|_F\cr
&=2\|\mathcal{P}_{\mathcal{S}_l}\mathcal{G}(\bm{y}_{l}+p^{-1}\mathcal{P}_{\Omega}(\bm{y}-\bm{y}_l))-\mathcal{G}\bm{y}\|_F\cr
&\leq 2\|\mathcal{P}_{\mathcal{S}_l}\mathcal{G}\bm{y}-\mathcal{G}\bm{y}\|_F+2\|(\mathcal{P}_{\mathcal{S}_l}\mathcal{G}-p^{-1}\mathcal{P}_{\mathcal{S}_l}\mathcal{G}\mathcal{P}_{\Omega})(\bm{y}_l-\bm{y}) \|_F\cr
&=2\|(\mathcal{I}-\mathcal{P}_{\mathcal{S}_l})(\bm{L}_l-\mathcal{G}\bm{y})\|_F+2\|(\mathcal{P}_{\mathcal{S}_l}\mathcal{G}\mathcal{G}^*-p^{-1}\mathcal{P}_{\mathcal{S}_l}\mathcal{G}\mathcal{P}_{\Omega}\mathcal{G}^*)(\bm{L}_{l}-\mathcal{G}\bm{y})\|_F\cr
&\leq 2\|(\mathcal{I}-\mathcal{P}_{\mathcal{S}_l})(\bm{L}_l-\mathcal{G}\bm{y})\|_F+2\|(\mathcal{P}_{\mathcal{S}_l}\mathcal{G}\mathcal{G}^*\mathcal{P}_{\mathcal{S}_l}-p^{-1}\mathcal{P}_{\mathcal{S}_l}\mathcal{G}\mathcal{P}_{\Omega}\mathcal{G}^*\mathcal{P}_{\mathcal{S}_l})(\bm{L}_l-\mathcal{G}\bm{y})\|_F\cr
&\quad+2\|\mathcal{P}_{\mathcal{S}_l}\mathcal{GG}^*(\mathcal{I}-\mathcal{P}_{\mathcal{S}_l})(\bm{L}_l-\mathcal{G}\bm{y})\|_F
+2p^{-1}\|\mathcal{P}_{\mathcal{S}_l}\mathcal{G}\mathcal{P}_{\Omega}\mathcal{G}^*(\mathcal{I}-\mathcal{P}_{\mathcal{S}_l})(\bm{L}_l-\mathcal{G}\bm{y})\|_F,\\
&:=I_1+I_2+I_3+I_4,
\end{align*}
where the second inequality comes from the fact that $\bm{L}_{l+1}$ is the best rank $r$ approximation to $\bm{W}_{l}$, the second equality follows from $(\I-\P_{\S_l})\BL_l=0$, $\by_l=\G^*\BL_l$ and $\G^*\G=\I$.

Let us first assume \eqref{c4} holds. Then the application of Lem.~\ref{lem:Riemannian} gives 
\begin{align*}
I_1+I_3+I_4&\leq \left(\frac{4\|\bm{L}_l-\mathcal{G}\bm{y}\|_F}{\sigma_{\min}(\mathcal{G}\bm{y})}+2p^{-1}\|\mathcal{P}_{\Omega}\mathcal{G}^{*}\mathcal{P}_{\mathcal{S}_l}\| \frac{\|\bm{L}_l-\mathcal{G}\bm{y}\|_F}{\sigma_{\min}(\mathcal{G}\bm{y})}\right)\|\bm{L}_l-\mathcal{G}\bm{y}\|_F\cr
&\leq 2\vep_0 \|\bm{L}_l-\mathcal{G}\bm{y}\|_F,
\end{align*}
where the last inequality follows from \eqref{c2}, \eqref{eq:POGPS} and the fact $\ln\P_{\S_l}\G\P_\Omega\rn=\ln\P_\Omega\G^*\P_{\S_l}\rn$. Moreover, \eqref{eq:RIPPL} implies 
\begin{align*}
I_2\leq8\vep_0 \|\bm{L}_l-\mathcal{G}\bm{y}\|_F.
\end{align*}
Therefore putting the bounds for $I_1,~I_2,~I_3, \mbox{ and }I_4$ together gives 
\begin{align*}
\|\bm{L}_{l+1}-\mathcal{G}\bm{y}\|_F\leq \nu \|\bm{L}_l-\mathcal{G}\bm{y}\|_F,
\end{align*}
where $\nu=10\vep_0<1$. Since \eqref{c4} holds for $l=0$ by the assumption of Thm.~\ref{thm:local} and 
$\ln \bm{L}_l-\mathcal{G}\bm{y}\rn_F$ is a contractive sequence,  \eqref{c4} holds for all $l\geq 0$. Thus 
\begin{align*}
\ln\by_l-\by\rn=\ln\G^*(\BL_l-\G\by)\rn\leq \ln \BL_l-\G\by\rn_F\leq\nu^l\ln \BL_0-\G\by\rn_F,
\end{align*}
where we have utilized the facts $\by_l=\G^*\BL_l$, $\G^*\G=\I$ and $\ln\G^*\rn\leq 1$.
\end{proof}
\subsection{Proofs of Lemma~\ref{lem:initial} and Theorem~\ref{thm:IHT}}\label{sec:proofinit1}
\begin{proof}[Proof of Lemma~\ref{lem:initial}]
Recall that $\BL_0 = \T_r(p^{-1}\H\P_\Omega(\bx))=\T_r(p^{-1}\G\P_\Omega(\by))$ and $\H\bx=\G\by$. Let us first bound $\ln p^{-1}\G\P_\Omega(\by)-\G\by \rn$. Since $p=\frac{m}{n}$, we have 
\begin{align*}
p^{-1}\G\P_\Omega(\by)-\G\by&=\sum_{k=1}^m\lb\frac{n}{m}y_{a_k}\BH_{a_k}-\frac{1}{m}\G\by\rb:=\sum_{k=1}^m\BZ_{a_k}.
\end{align*}
Because each $a_{k}$ is drawn uniformly from $\lcb 0,\cdots,n-1\rcb$, it is trivial that $\E\lb\BZ_{a_k}\rb=0$. Moreover, we have 
\begin{align*}
\E\lb \BZ_{a_k}\BZ_{a_k}^*\rb&=\E\lb \frac{n^2}{m^2}|y_{a_k}|^2\BH_{a_k}\BH_{a_k}^*\rb-\frac{1}{m^2}(\G\by)(\G\by)^*
\\&=\frac{n}{m^2}\sum_{a=0}^{n-1}|y_a|^2\BH_{a}\BH_{a}^*-\frac{1}{m^2}(\G\by)(\G\by)^*\\
&=\frac{n}{m^2}\BC-\frac{1}{m^2}(\G\by)(\G\by)^*,
\end{align*}
where $\BC$ is a diagonal matrix which corresponds to the diagonal part of $(\G\by)(\G\by)^*$. Therefore
\begin{align*}
\ln \E\lb \sum_{k=1}^m\BZ_{a_k}\BZ_{a_k}^*\rb\rn&\leq\max\lcb\frac{n}{m}\ln\BC\rn,\frac{1}{m}\ln (\G\by)(\G\by)^*\rn\rcb\\
&\leq \frac{n}{m}\ln\G\by\rn^2_{2\rightarrow\infty},
\end{align*}
where $\ln\G\by\rn_{2\rightarrow\infty}$ denotes the maximum row $\ell_2$ norm of $\G\by$. Similarly we can get 
\begin{align*}
\ln \E\lb \sum_{k=1}^m\BZ_{a_k}^*\BZ_{a_k}\rb\rn\leq\frac{n}{m}\ln(\G\by)^*\rn^2_{2\rightarrow\infty}.
\end{align*}
The definition of $\BH_a$ in \eqref{eq:Ha} implies $\ln\BH_a\rn\leq\frac{1}{\sqrt{w_a}}$. So
\begin{align*}
\ln\BZ_{a_k}\rn\leq \frac{n}{m}|y_{a_k}|\ln\BH_{a_k}\rn+\frac{1}{m}\sum_{a=0}^{n-1}|y_a|\ln\BH_a\rn\leq\frac{2n}{m}\ln\D^{-1}\by\rn_\infty.
\end{align*}
By matrix Bernstein inequality in Lem.~\ref{lem:bernstein}, one can  show that there exists a universal constant $C>0$  such that
\begin{align*}
\ln\sum_{k=1}^m\BZ_{a_k}\rn\leq C\lb\sqrt{\frac{n\log(n)}{m}}\max\lcb \ln\G\by\rn_{2\rightarrow\infty},\ln(\G\by)^*\rn_{2\rightarrow\infty}\rcb+\frac{n\log(n)}{m}\ln\D^{-1}\by\rn_\infty\rb
\end{align*}
with probability at least $1-n^{-2}$. Consequently 
 on the same event we have 
\begin{align*}
\ln\BL_0-\G\by\rn&\leq \ln\BL_0-p^{-1}\G\P_\Omega(\by)\rn+\ln p^{-1}\G\P_\Omega(\by)-\G\by\rn\leq 2\ln p^{-1}\G\P_\Omega(\by)-\G\by\rn\\
&\leq  C\lb\sqrt{\frac{n\log(n)}{m}}\max\lcb \ln\G\by\rn_{2\rightarrow\infty},\ln(\G\by)^*\rn_{2\rightarrow\infty}\rcb+\frac{n\log(n)}{m}\ln\D^{-1}\by\rn_\infty\rb.\numberthis\label{eq:initial_error}
\end{align*}
Thus it only remains to bound $\max\lcb \ln\G\by\rn_{2\rightarrow\infty},\ln(\G\by)^*\rn_{2\rightarrow\infty}\rcb$
and $\ln\D^{-1}\by\rn_\infty$ in terms of $\ln\G\by\rn$. From $\G\by=\mathcal{H}\bm{x}=\bm{U}\bm{\Sigma}\bm{V}^*=\bm{E}_L\bm{D}\bm{E}_R^T$, we get
\begin{align*}
\ln\G\by\rn_{2\rightarrow\infty}^2&=\max_i\|\bm{e}_i^*(\mathcal{G}\bm{y})\|^2=\max_i\|\bm{e}_i^*\bm{U}\bm{\Sigma}\bm{V}^*\|^2
\leq\max_i\|\bm{e}_i^*\bm{U}\|^2\|\bm{\Sigma}\|^2 \\
&=\max_i\ln\BU^{(i,:)}\rn^2\|\mathcal{G}\bm{y}\|_2^2
\leq \frac{\mu_0c_sr}{n}\|\mathcal{G}\bm{y}\|_2^2,\numberthis\label{eq:bd_gy}
\end{align*}
where the last inequality follows from Lem.~\ref{lem:U_V}. Similarly we also have
\begin{equation}\label{eq:bd_gyt}
\ln(\G\by)^*\rn^2_{2\rightarrow\infty}\leq \frac{\mu_0c_sr}{n}\|\mathcal{G}\bm{y}\|_2^2.
\end{equation}

The infinity norm of $\D^{-1}\by$ can be bounded as follows 
\begin{align*}
\ln \D^{-1}\by\rn_\infty&=\ln\G\by\rn_\infty=\max_{i,j}|\be_i^*(\G\by)\be_j|\leq \max_{i,j}\ln\be_i^*\BE_L\rn\ln\BD\rn\ln\BE_R^T\be_j\rn\\
&\leq r\ln\BD\rn\leq r\ln\BE_L^\dag\rn\ln\G\by\rn\ln(\BE_R^T)^\dag\rn\leq \frac{\mu_0c_sr}{n}\ln\G\by\rn,\numberthis\label{eq:bd_dy}
\end{align*}
where the last inequality follows from the $\mu_0$-incoherence of $\G\by$. 

Finally inserting \eqref{eq:bd_gy}, \eqref{eq:bd_gyt} and \eqref{eq:bd_dy} into \eqref{eq:initial_error} gives
\begin{align*}
\ln\BL_0-\G\by\rn&\leq C\sqrt{\frac{\mu_0c_sr\log(n)}{m}}\ln\G\by\rn
\end{align*}
provided $m\geq \mu_0c_sr\log(n)$.
\end{proof}
\begin{proof}[Proof of Theorem~\ref{thm:IHT}]
Following from \eqref{eq:err_x_y}, 
we only need to verify when the three conditions in Thm.~\ref{thm:local} are satisfied. Lemma~\ref{lem:sampling} implies \eqref{c1} holds with probability at least $1-n^{-2}$. Lemmas~\ref{lem:U_V} and \ref{lem:RIPPS} guarantees \eqref{c2} is true with probability at least $1-n^{-2}$ if $m\geq C\vep_0^{-2}\mu_0c_sr\log(n)$ for a sufficiently large numerical constant $C>0$.  Similarly  \eqref{c3} can be satisfied with probability at least $1-n^2$ if $m\geq C(1+\vep_0)\vep_0^{-1}\mu_0^{1/2}c_s^{1/2}\kappa rn^{1/2}\log^{3/2}(n)$ following Lem.~\ref{lem:initial} and the fact $\ln\BL_0-\G\by\rn_F\leq\sqrt{2r}\ln\BL_0-\G\by\rn$, where $\kappa$ denotes the condition number of $\G\by$. Taking an upper bound on the number of measurements completes the proof of Thm.~\ref{thm:IHT}.
\end{proof}
\subsection{Proofs of Lemma~\ref{lem:resampling} and Theorem~\ref{thm:resampling}}
The proof of Lem.~\ref{lem:resampling} relies on the following estimation of $\Big\|\mathcal{P}_{\widehat{\mathcal{S}}_l}\mathcal{G}\left(\widehat{p}^{-1}\mathcal{P}_{\widehat{\Omega}_{l+1}}-\mathcal{I}\right)\mathcal{G}^*\left(\mathcal{P}_{\BU}-\mathcal{P}_{\widehat{\BU}_l}\right)\Big\|$, which is a generalization of the asymmetric restricted isometry property \cite{Completion} from matrix completion to low rank Hankel matrix completion.
\begin{lemma}\label{lem:incoherence}
Assume there exists a numerical constant $\mu$ such that
\begin{equation}\label{eq:inc1}
\|\mathcal{P}_{\widehat{\bm{U}}_l}\BH_a\|_F^2
\leq \frac{\mu c_s r}{n},\quad \|\mathcal{P}_{\widehat{\BV}_l}\BH_a\|_F^2
\leq \frac{\mu c_s r}{n},
\end{equation}
and
\begin{equation}\label{eq:inc2}
\|\mathcal{P}_{\BU}\BH_a\|_F^2
\leq \frac{\mu c_s r}{n},\quad\|\mathcal{P}_{\BV}\BH_a\|_F^2
\leq \frac{\mu c_s r}{n}.
\end{equation}
for all $0\leq a\leq n-1$. Let $\widehat{\Omega}_{l+1}=\lcb a_k~|~k=1,\cdots,\widehat{m}\rcb$ be a set of indices sampled with replacement.  If $\mathcal{P}_{\widehat{\Omega}_{l+1}}$ is independent of $\BU$, $\BV$, $\widehat{\bm{U}}_l$  and $\widehat{\BV}_l$, then 
$$
\Big\|\mathcal{P}_{\widehat{\mathcal{S}}_l}\mathcal{G}\left(\mathcal{I}-\widehat{p}^{-1}\mathcal{P}_{\widehat{\Omega}_{l+1}}\right)\mathcal{G}^*\left(\mathcal{P}_{\BU}-\mathcal{P}_{\widehat{\BU}_l}\right)\Big\|
\leq \sqrt{\frac{160\mu c_s r\log(n)}{\widehat{m}}}
$$
with probability at least $1-n^{-2}$ provided
$$
\widehat{m}\geq \frac{125}{18}\mu c_s r\log(n).
$$
\end{lemma}
\begin{proof}
Since for any $\BZ\in\C^{n_1\times n_2}$
\begin{align*}
\mathcal{P}_{\widehat{\mathcal{S}}_l}\mathcal{G}\mathcal{P}_{\widehat{\Omega}_{l+1}}\mathcal{G}^*\left(\mathcal{P}_{\BU}-\mathcal{P}_{\widehat{\BU}_l}\right)(\BZ) = \sum_{k=1}^\whm\la\BZ,\left(\mathcal{P}_{\BU}-\mathcal{P}_{\widehat{\BU}_l}\right)(\BH_{a_k})\ra\P_{\S_l}(\BH_{a_k}), 
\end{align*}
we can rewrite $\mathcal{P}_{\widehat{\mathcal{S}}_l}\mathcal{G}\mathcal{P}_{\widehat{\Omega}_{l+1}}\mathcal{G}^*\left(\mathcal{P}_{\BU}-\mathcal{P}_{\widehat{\BU}_l}\right)$ as 
\begin{align*}
\mathcal{P}_{\widehat{\mathcal{S}}_l}\mathcal{G}\mathcal{P}_{\widehat{\Omega}_{l+1}}\mathcal{G}^*\left(\mathcal{P}_{\BU}-\mathcal{P}_{\widehat{\BU}_l}\right) = \sum_{k=1}^\whm\P_{\S_l}(\BH_{a_k})\otimes \left(\mathcal{P}_{\BU}-\mathcal{P}_{\widehat{\BU}_l}\right)(\BH_{a_k}).
\end{align*}
Define the random operator 
\begin{align*}
\R_{a_k} = \P_{\widehat{\S}_l}(\BH_{a_k})\otimes \left(\mathcal{P}_{\BU}-\mathcal{P}_{\widehat{\BU}_l}\right)(\BH_{a_k})-\frac{1}{n}
\mathcal{P}_{\widehat{\mathcal{S}}_l}\mathcal{G}\mathcal{G}^*\left(\mathcal{P}_{\BU}-\mathcal{P}_{\widehat{\BU}_l}\right).
\end{align*}
Then it is easy to see that $\E\lb\R_{a_k}\rb=0$. By assumption, for any $0\leq a\leq n-1$, 
\begin{align*}
\|\mathcal{P}_{\widehat{\mathcal{S}}_l}\lb\bm{H}_a\rb\|_F^2\leq \|\mathcal{P}_{\widehat{U}_l}\lb\bm{H}_a\rb\|_F^2+\|\mathcal{P}_{\widehat{V}_l}\lb\bm{H}_a\rb\|_F^2
\leq \frac{2\mu c_s r}{n}.
\end{align*}
So
\begin{align*}
\ln\R_{a_k}\rn&\leq \ln \mathcal{P}_{\widehat{\mathcal{S}}_l}\lb\bm{H}_{a_k}\rb\rn_F\ln\lb \mathcal{P}_{\BU}-\mathcal{P}_{\widehat{\BU}_l}\right)(\BH_{a_k})\rn_F+\frac{1}{n}\ln\mathcal{P}_{\widehat{\mathcal{S}}_l}\mathcal{G}\mathcal{G}^*\left(\mathcal{P}_{\BU}-\mathcal{P}_{\widehat{\BU}_l}\right)\rn\leq\frac{5\mu c_sr}{n}.
\end{align*}
Next let us bound $\ln\mathds{E}(\mathcal{R}_{a_k}\mathcal{R}_{a_k}^*)\rn$ as follows 
\begin{align*}
\ln\mathds{E}(\mathcal{R}_{a_k}\mathcal{R}_{a_k}^*)\rn&=\ln\E\lb\ln \lb \mathcal{P}_{\BU}-\mathcal{P}_{\widehat{\BU}_l}\right)(\BH_{a_k})\rn_F^2\mathcal{P}_{\widehat{\mathcal{S}}_l}\lb\bm{H}_{a_k}\rb\otimes \mathcal{P}_{\widehat{\mathcal{S}}_l}\lb\bm{H}_{a_k}\rb\rb-\frac{1}{n^2}\mathcal{P}_{\widehat{\mathcal{S}}_l}\mathcal{G}\mathcal{G}^*\left(\mathcal{P}_{\BU}-\mathcal{P}_{\widehat{\BU}_l}\right)^2\G\G^*\P_{\widehat{\S}_l}\rn\\
&\leq \ln \E\lb\ln \lb \mathcal{P}_{\BU}-\mathcal{P}_{\widehat{\BU}_l}\right)(\BH_{a_k})\rn_F^2\mathcal{P}_{\widehat{\mathcal{S}}_l}\lb\bm{H}_{a_k}\rb\otimes \mathcal{P}_{\widehat{\mathcal{S}}_l}\lb\bm{H}_{a_k}\rb\rb\rn+\frac{4}{n^2}\\
&\leq \frac{4\mu c_sr}{n}\ln\E\lb\mathcal{P}_{\widehat{\mathcal{S}}_l}\lb\bm{H}_{a_k}\rb\otimes \mathcal{P}_{\widehat{\mathcal{S}}_l}\lb\bm{H}_{a_k}\rb\rb\rn+\frac{4}{n^2}\\
&= \frac{4\mu c_sr}{n^2}\ln\mathcal{P}_{\widehat{\mathcal{S}}_l}\G\G^*\mathcal{P}_{\widehat{\mathcal{S}}_l}\rn+\frac{4}{n^2}\\
&\leq \frac{8\mu c_sr}{n^2}.
\end{align*}
This implies 
\begin{align*}
\ln\E\lb\sum_{k=1}^\whm\R_{a_k}\R_{a_k}^*\rb\rn\leq\sum_{k=1}^\whm\ln \mathds{E}(\mathcal{R}_{a_k}\mathcal{R}_{a_k}^*)\rn\leq\frac{8\mu c_sr\whm}{n^2}.
\end{align*}
We can similarly obtain 
\begin{align*}
\ln\E\lb\sum_{k=1}^\whm\R_{a_k}^*\R_{a_k}\rb\rn\leq\frac{12\mu c_sr\whm}{n^2}.
\end{align*}
So the application of the matrix Bernstein inequality  in Lem.~\ref{lem:bernstein} gives 
\begin{align*}
\mathbb{P}\left\{\left\|\sum_{k=1}^{\widehat{m}}\mathcal{R}_{a_k}\right\|\geq t\right\}
\leq
2n_1n_2\exp\left(\frac{-t^2/2}{\frac{12\mu c_s {\widehat{m}} r}{n^2}+\frac{5\mu c_s r}{n}t/3}\right).
\end{align*}
If $t\leq \frac{24{\widehat{m}}}{5n}$, then
$$
\mathbb{P}\left\{\left\|\sum_{k=1}^{\widehat{m}}\mathcal{R}_{a_k}\right\|\geq t\right\}
\leq 2n_1n_2\mathrm{exp}\left(\frac{-t^2/2}{\frac{20\mu c_s{\widehat{m}} r}{n^2}}\right)
\leq n^2\mathrm{exp}\left(\frac{-t^2/2}{\frac{20\mu c_s{\widehat{m}} r}{n^2}}\right).
$$
Setting $t=\sqrt{\frac{160\mu c_s{\widehat{m}} r\log(n)}{n^2}}$ gives
$$
\mathbb{P}\left\{\left\|\sum_{k=1}^{\widehat{m}}\mathcal{R}_{a_k}\right\|\geq t\right\}
\leq n^{-2}.
$$
The condition $t\leq\frac{24{\widehat{m}}}{5n}$ implies
$
{\widehat{m}}\geq\frac{125}{18}\mu c_sr\log(n).
$
The proof is complete because
$$
\frac{n}{\widehat{m}}\sum_{k=1}^{\widehat{m}}\mathcal{R}_{a_k}=\mathcal{P}_{\widehat{\mathcal{S}}_l}\mathcal{G}\left(\widehat{p}^{-1}\mathcal{P}_{\widehat{\Omega}_{l+1}}-\mathcal{I}\right)\mathcal{G}^*\left(\mathcal{P}_{\BU}-\mathcal{P}_{\widehat{\BU}_l}\right).$$
\end{proof}
The following lemma from \cite{Completion} will also be used in the proof of Lem.~\ref{lem:resampling}.
\begin{lemma}\label{lem:trim_not_increase}
Let $\widetilde{\BL}_l=\widetilde{\BU}_l\widetilde{\BS}_l\widetilde{\BV}^*_l$  and $\G\by=\BU\BS\BV^*$ be two rank $r$ matrices which satisfy
\begin{align*}
\ln\widetilde{\BL}_l-\G\by\rn_F\leq \frac{\sigma_{\min}(\G\by)}{10\sqrt{2}}.
\end{align*}
Assume 
$
\ln\BU^{(i,:)}\rn^2\leq\frac{\mu_0c_sr}{n}\mbox{ and }\ln\BV^{(j,:)}\rn^2\leq\frac{\mu_0c_sr}{n}.
$
Then the matrix  $\widehat{\BL}_l=\mathrm{Trim}_{\mu_0}(\widetilde{\BL}_l)=\widehat{\BU}_l\widehat{\BS}_l\widehat{\BV}^*_l$ returned by Alg.~\ref{Trimming}
satisfies 
\begin{align*}
\ln\widehat{\BL}_l-\G\by\rn_F\leq 8\kappa\ln\widetilde{\BL}_l-\G\by\rn_F\quad\mbox{and}\quad\max\lcb\ln\widehat{\BU}^{(i,:)}\rn^2,\ln\widehat{\BV}^{(j,:)}\rn^2\rcb\leq\frac{100\mu_0c_sr}{81n},
\end{align*}
where $\kappa$ denotes the condition number of $\G\by$.
\end{lemma}
\begin{proof}[Proof of Lemma~\ref{lem:resampling}]
Let us first assume that 
\begin{align*}
\ln\widetilde{\bm{L}}_{l}-\mathcal{G}\bm{y}\rn_F\leq \dfrac{\sigma_{\min}(\mathcal{G}\bm{y})}{256\kappa^2}.\numberthis\label{lem10_assumption}
\end{align*}
Then the application of  Lem.~\ref{lem:trim_not_increase} implies that
\begin{align*}
\ln\widehat{\BL}_l-\G\by\rn_F\leq 8\kappa\ln\widetilde{\BL}_l-\G\by\rn_F\quad\mbox{and}\quad\max\lcb\ln\widehat{\BU}^{(i,:)}\rn^2,\ln\widehat{\BV}^{(j,:)}\rn^2\rcb\leq\frac{100\mu_0c_sr}{81n}\numberthis\label{eq:after_trimming}
\end{align*}
by noting that $\ln\BU^{(i,:)}\rn^2\leq\frac{\mu_0c_sr}{n}\mbox{ and }\ln\BV^{(j,:)}\rn^2\leq\frac{\mu_0c_sr}{n}$ following from Lem.~\ref{lem:U_V}. Moreover, direct calculation gives
\begin{align*}
\left\|\mathcal{P}_{\widehat{\bm{U}}_{l}}\bm{H}_a\right\|_F^2
=\left\|\widehat{\bm{U}}_{l}^*\bm{H}_a\right\|_F^2
=\frac{1}{|\Gamma_a|}\sum_{i\in \Gamma_a}\left\|\left(\widehat{\bm{U}}_{l}\right)^{(i,:)}\right\|_2^2
\leq\frac{100\mu_0 c_s r}{81n},\numberthis\label{eq:U_hat_incoherence}
\end{align*}
where $\Gamma_a$ is the set of row indices for non-zero entries in $\bm{H}_a$ with cardinality $|\Gamma_a|=w_a$.  Similarly,
\begin{align*}
\left\|\mathcal{P}_{\widehat{\bm{V}}_{l}}\bm{H}_a\right\|_F^2
\leq\frac{100\mu_0 c_s r}{81n}.\numberthis\label{eq:V_hat_incoherence}
\end{align*}

Recall that $\by=\D\bx$ and $\G\by=\H\bx$.  Define $\widehat{\by}_l=\D\widehat{\bx}_l$.  Then  $\widehat{\by}_l=\G^*\widehat{\BL}_l$ and
\begin{align*}
\mathcal{P}_{\widehat{\mathcal{S}}_l}\mathcal{H}\left(\widehat{\bm{x}}_l+\widehat{p}^{-1}\mathcal{P}_{\Omega_{l+1}}\left(\bm{x}-\widehat{\bm{x}}_l\right)\right)=\P_{\S_l}\G\lb\widehat{\by}_l+\widehat{p}^{-1}\mathcal{P}_{\Omega_{l+1}}\left(\bm{y}-\widehat{\bm{y}}_l\right)\rb.
\end{align*}
Consequently,
\begin{equation*}
\begin{split}
\|\widetilde{\bm{L}}_{l+1}-\mathcal{G}\bm{y}\|_F 
&\leq 2\ln\mathcal{P}_{\mathcal{S}_l}\mathcal{G}\lb\widehat{\bm{y}}_{l}+\widehat{p}^{-1}\mathcal{P}_{\widehat{\Omega}_{l+1}}\lb\bm{y}-\widehat{\bm{y}}_l\rb\rb-\mathcal{G}\bm{y}\rn_F\cr
&\leq 2\ln\mathcal{P}_{\widehat{\mathcal{S}}_l}\mathcal{G}\bm{y}-\mathcal{G}\bm{y}\rn_F+2\ln\lb\mathcal{P}_{\widehat{\mathcal{S}}_l}\mathcal{G}-\widehat{p}^{-1}\mathcal{P}_{\widehat{\mathcal{S}}_l}\mathcal{G}\mathcal{P}_{\widehat{\Omega}_{l+1}}\rb\lb\widehat{\bm{y}}_l-\bm{y}\rb \rn_F\cr
&=2\ln\lb\mathcal{I}-\mathcal{P}_{\widehat{\mathcal{S}}_l}\rb\mathcal{G}\bm{y}\rn_F+2\ln\lb\mathcal{P}_{\widehat{\mathcal{S}}_l}\mathcal{G}\mathcal{G}^*-\widehat{p}^{-1}\mathcal{P}_{\widehat{\mathcal{S}}_l}\mathcal{G}\mathcal{P}_{\widehat{\Omega}_{l+1}}\mathcal{G}^*\rb\lb\widehat{\bm{L}}_{l}-\mathcal{G}\bm{y}\rb\rn_F\cr
&\leq 2\ln\lb\mathcal{I}-\mathcal{P}_{\widehat{\mathcal{S}}_l}\rb\lb\widehat{\bm{L}}_l-\mathcal{G}\bm{y}\rb\rn_F+2\ln\lb\mathcal{P}_{\widehat{\mathcal{S}}_l}\mathcal{G}\mathcal{G}^*\mathcal{P}_{\widehat{\mathcal{S}}_l}-\widehat{p}^{-1}\mathcal{P}_{\widehat{\mathcal{S}}_l}\mathcal{G}\mathcal{P}_{\widehat{\Omega}_{l+1}}\mathcal{G}^*\mathcal{P}_{\widehat{\mathcal{S}}_l}\rb\lb\widehat{\bm{L}}_l-\mathcal{G}\bm{y}\rb\rn_F\cr
&+2\ln\mathcal{P}_{\widehat{\mathcal{S}}_l}\mathcal{G}\lb\mathcal{I}-\widehat{p}^{-1}\mathcal{P}_{\widehat{\Omega}_{l+1}}\rb\mathcal{G}^*\left(\mathcal{I}-\mathcal{P}_{\widehat{\mathcal{S}}_l}\right)(\widehat{\bm{L}}_l-\mathcal{G}\bm{y})\rn_F\\
&:=I_5+I_6+I_7.\end{split}
\end{equation*}
The first item $I_5$ can be bounded as 
\begin{align*}
I_5\leq \frac{2\ln\widehat{\bm{L}}_l-\mathcal{G}\bm{y}\rn_F^2}{\sigma_{\min}(\mathcal{G}\bm{y})}\leq\frac{1}{2}\ln\widetilde{\bm{L}}_{l}-\mathcal{G}\bm{y}\rn_F,
\end{align*}
which follows from Lem.~\ref{lem:Riemannian}, the left inequality of \eqref{eq:after_trimming} and the assumption \eqref{lem10_assumption}. The application of Lem.~\ref{lem:RIPPS} together with \eqref{eq:U_hat_incoherence} and \eqref{eq:V_hat_incoherence} implies
\begin{align*}
I_6&\leq2\sqrt{\frac{3200\mu_0 c_sr\log(n)}{81\whm}}\ln\widehat{\bm{L}}_l-\mathcal{G}\bm{y}\rn_F\leq 16\kappa\sqrt{\frac{3200\mu_0 c_sr\log(n)}{81\whm}}\ln\widetilde{\bm{L}}_l-\mathcal{G}\bm{y}\rn_F
\end{align*}
with probability at least $1-n^2$.
To bound $I_7$, first note that 
\begin{equation*}
\begin{split}
\lb\mathcal{I}-\mathcal{P}_{\widehat{\mathcal{S}}_l}\rb\lb\widehat{\bm{L}}_l-\mathcal{G}\bm{y}\rb&=\lb\mathcal{I}-\mathcal{P}_{\widehat{\mathcal{S}}_l}\rb\lb-\mathcal{G}\bm{y}\rb
=\lb\bm{I}-\widehat{\bm{U}}_l\widehat{\bm{U}}_l^*\rb\lb-\mathcal{G}\bm{y}\rb\lb\bm{I}-\widehat{\bm{V}}_l\widehat{\bm{V}}_l^*\rb\cr
&=\lb\bm{U}\bm{U}^*-\widehat{\bm{U}}_l\widehat{\bm{U}}_l^*\rb\lb\widehat{\bm{L}}_l-\mathcal{G}\bm{y}\rb\lb\bm{I}-\widehat{\bm{V}}_l\widehat{\bm{V}}_l^*\rb\\
&=\lb\mathcal{P}_{\bm{U}}-\mathcal{P}_{\widehat{\bm{U}}_l}\rb\lb\mathcal{I}-\mathcal{P}_{\bm{V}}\rb\lb \widehat{\bm{L}}_l-\mathcal{G}\bm{y}\rb.
\end{split}
\end{equation*}
Therefore 
\begin{equation*}
\begin{split}
I_7&=2\ln\mathcal{P}_{\widehat{\mathcal{S}}_l}\mathcal{G}\lb\mathcal{I}-\widehat{p}^{-1}\mathcal{P}_{\widehat{\Omega}_{l+1}}\rb\mathcal{G}^*\left(\mathcal{I}-\mathcal{P}_{\widehat{\mathcal{S}}_l}\right)\lb\mathcal{P}_{\bm{U}}-\mathcal{P}_{\widehat{\bm{U}}_l}\rb\lb\mathcal{I}-\mathcal{P}_{\bm{V}}\rb\lb\widehat{\bm{L}}_l-\mathcal{G}\bm{y}\rb\rn_F\cr
&\leq 2\left\|\mathcal{P}_{\widehat{\mathcal{S}}_l}\mathcal{G}\lb\mathcal{I}-\widehat{p}^{-1}\mathcal{P}_{\widehat{\Omega}_{l+1}}\rb\mathcal{G}^*\left(\mathcal{I}-\mathcal{P}_{\widehat{\mathcal{S}}_l}\right)\lb\mathcal{P}_{\bm{U}}-\mathcal{P}_{\widehat{\bm{U}}_l}\rb\right\| \left\|\widehat{\bm{L}}_l-\mathcal{G}\bm{y}\right\|_F\\
&\leq 16\kappa\sqrt{\frac{16000\mu_0 c_s r\log(n)}{81\widehat{m}}}\left\|\widetilde{\bm{L}}_l-\mathcal{G}\bm{y}\right\|_F
\end{split}
\end{equation*}
with probability at least $1-n^2$,
where the last inequality follows from Lem.~\ref{lem:incoherence} and the left inequality of \eqref{eq:after_trimming}.
Putting the bounds for $I_5$, $I_6$ and $I_7$ together gives 
\begin{align*}
\|\widetilde{\bm{L}}_{l+1}-\mathcal{G}\bm{y}\|_F \leq\lb\frac{1}{2}+326\kappa\sqrt{\frac{\mu_0 c_s r\log(n)}{\widehat{m}}}\rb\left\|\widetilde{\bm{L}}_l-\mathcal{G}\bm{y}\right\|_F\leq \frac{5}{6}\left\|\widetilde{\bm{L}}_l-\mathcal{G}\bm{y}\right\|_F
\end{align*}
with probability at least $1-2n^{-2}$ provided $\whm\geq C\mu_0 c_s\kappa^2 r\log(n)$ for a sufficiently large universal constant $C$. Clearly on the same event, \eqref{lem10_assumption} also holds for the $(l+1)$-th iteration.

Since $\widetilde{\bm{L}}_0=\mathcal{T}_r\lb\widehat{p}^{-1}\mathcal{H}\mathcal{P}_{\Omega_0}\lb\bm{x}\rb\rb$, \eqref{lem10_assumption} is valid for $l=0$ with probability at least $1-n^2$ provides 
\begin{align*}
\whm\geq C\mu_0c_s\kappa^6r^2\log(n)
\end{align*}
for some numerical constant $C>0$. Taking the upper bound on the number of measurements completes the proof of Lem.~\ref{lem:resampling} by noting $\H\bx=\G\by$.
\end{proof}
\begin{proof}[Proof of Theorem~\ref{thm:resampling}]
The third condition \eqref{c3} in Thm.~\ref{thm:local} can be satisfied with probability at least $1-(2L+1)n^{-2}$ if we take $L=\left\lceil6\log\left(\frac{\sqrt{n}\log(n)}{16\varepsilon_0}\right)\right\rceil$. So the theorem can be proved by combining this result together with Lems.~\ref{lem:sampling} and \ref{lem:RIPPS}.
\end{proof}
\section{Conclusion and Future Directions}\label{sec:conclusion}
We have proposed two new algorithms IHT and FIHT to reconstruct  spectrally sparse signals from partial revealed entries via low rank Hankel matrix completion. While the empirical phase transitions of IHT and FIHT are similar to those of existing convex optimization approaches in the literature, IHT and FIHT are more computationally efficient. Theoretical recovery guarantees are established for FIHT under two different initialization strategies. The sampling complexity for FIHT with one step hard thresholding initialization is highly pessimistic when compared with the empirical observations, which suggests the possibility of improving this result in the future.

Though IHT and FIHT are implemented for fixed rank problems (i.e., the number of frequencies in the spectrally sparse signal is known a priori) in this paper, the common rank increasing or decreasing heuristics can be incorporated into them as well. When the number of frequencies is not known, we also suggest replacing the hard thresholding operator in IHT and FIHT with the soft thresholding operator or more complicated shrinkage operators. The theoretical guarantee analysis of these new variants is an interesting future  topic.

The numerical simulations in Sec.~\ref{sec:robust} show that both IHT and FIHT are very robust under additive noise.  As future work, we will extend our analysis to noisy measurements. It is also interesting to study  the Gaussian random sampling model for spectrally sparse signal reconstruction problems, and investigate whether  a new property analogous to D-RIP in \cite{candes2011compressed} can be  established for this model since low rank Hankel matrix reconstruction has similar algebraic structure with   compressed sensing under the tight frame analysis sparsity model as presented in Sec.~\ref{sec:connection}.
\section*{Acknowledgments}
KW acknowledges support from the NSF via grant {DTRA-DMS} 1322393.

\bibliography{hankel}

\begin{thebibliography}{10}

\bibitem{FISTA}
{\sc A.~Beck and M.~Teboulle}, {\em A fast iterative shrinkage-thresholding
  algorithm for linear inverse problems}, SIAM J. Imaging Sciences, 2 (2009),
  pp.~183--202.

\bibitem{CGIHT}
{\sc J.~Blanchard, J.~Tanner, and K.~Wei}, {\em {CGIHT}: {C}onjugate gradient
  iterative hard thresholding for compressed sensing and matrix completion},
  Information and Inference, 4 (2015), pp.~289--327.

\bibitem{bludav2009iht}
{\sc T.~Blumensath and M.~E. Davies}, {\em Iterative hard thresholding for
  compressed sensing}, Applied and Computational Harmonic Analysis, 27(3)
  (2009), pp.~265--274.

\bibitem{blumensathdavies2010niht}
\leavevmode\vrule height 2pt depth -1.6pt width 23pt, {\em Normalized iterative
  hard thresholding: {G}uaranteed stability and performance}, IEEE Journal of
  Selected Topics in Signal Processing, 4(2) (2010), pp.~298--309.

\bibitem{cai2008framelet}
{\sc J.-F. Cai, R.~H. Chan, and Z.~Shen}, {\em A framelet-based image
  inpainting algorithm}, Applied and Computational Harmonic Analysis, 24
  (2008), pp.~131--149.

\bibitem{PWGD}
{\sc J.-F. Cai, S.~Liu, and W.~Xu}, {\em A fast algorithm for reconstruction of
  spectrally sparse signals in super-resolution}, in SPIE Optical Engineering+
  Applications, International Society for Optics and Photonics, 2015,
  pp.~95970A--95970A.

\bibitem{CQXY:ACHA:16}
{\sc J.-F. Cai, X.~Qu, W.~Xu, and G.-B. Ye}, {\em Robust recovery of complex
  exponential signals from random gaussian projections via low rank {H}ankel
  matrix reconstruction}, Applied and Computational Harmonic Analysis,  (to
  appear).

\bibitem{cai2010framelet}
{\sc J.-F. Cai and Z.~Shen}, {\em Framelet based deconvolution}, J. Comput.
  Math, 28 (2010), pp.~289--308.

\bibitem{candes2011compressed}
{\sc E.~J. Candes, Y.~C. Eldar, D.~Needell, and P.~Randall}, {\em Compressed
  sensing with coherent and redundant dictionaries}, Applied and Computational
  Harmonic Analysis, 31 (2011), pp.~59--73.

\bibitem{candesrecht2009mc}
{\sc E.~J. Cand\`es and B.~Recht}, {\em Exact matrix completion via convex
  optimization}, Foundations of Computational Mathematics, 9(6) (2009),
  pp.~717--772.

\bibitem{CS}
{\sc E.~J. Cand{\`e}s, J.~Romberg, and T.~Tao}, {\em Robust uncertainty
  principles: {E}xact signal reconstruction from highly incomplete frequency
  information}, Information Theory, IEEE Transactions on, 52 (2006),
  pp.~489--509.

\bibitem{chan2003wavelet}
{\sc R.~H. Chan, T.~F. Chan, L.~Shen, and Z.~Shen}, {\em Wavelet algorithms for
  high-resolution image reconstruction}, SIAM Journal on Scientific Computing,
  24 (2003), pp.~1408--1432.

\bibitem{Chi}
{\sc Y.~Chen and Y.~Chi}, {\em Robust spectral compressed sensing via
  structured matrix completion}, Information Theory, IEEE Transactions on, 60
  (2014), pp.~6576--6601.

\bibitem{Mismatch}
{\sc Y.~Chi, L.~L. Scharf, A.~Pezeshki, and A.~R. Calderbank}, {\em Sensitivity
  to basis mismatch in compressed sensing}, Signal Processing, IEEE
  Transactions on, 59 (2011), pp.~2182--2195.

\bibitem{dong_shen_iciam}
{\sc B.~Dong and Z.~Shen}, {\em Image restoration: {A} data-driven
  perspective}, in Proceedings of the ICIAM, 2015.

\bibitem{donoho2006cs}
{\sc D.~L. Donoho}, {\em Compressed sensing}, IEEE Transactions on Information
  Theory, 52(4) (2006), pp.~1289--1306.

\bibitem{foucart2011htp}
{\sc S.~Foucart}, {\em Hard thresholding pursuit: {A}n algorithm for
  compressive sensing}, SIAM Journal on Numerical Analysis, 49 (2011),
  pp.~2543--2563.

\bibitem{goldfarbma2011fpca}
{\sc D.~Goldfarb and S.~Ma}, {\em Convergence of fixed-point continuation
  algorithms for matrix rank minimization}, Foundations of Computational
  Mathematics, 11(2) (2011), pp.~183--210.

\bibitem{CVX}
{\sc M.~Grant and S.~Boyd}, {\em {CVX}: Matlab software for disciplined convex
  programming, version 2.1}.
\newblock \url{http://cvxr.com/cvx}, Mar. 2014.

\bibitem{jmd2010svp}
{\sc P.~Jain, R.~Meka, and I.~Dhillon}, {\em Guaranteed rank minimization via
  singular value projection}, in Proceedings of the Neural Information
  Processing Systems Conference, 2010.

\bibitem{PROPACK}
{\sc R.~Larsen}, {\em {PROPACK}-software for large and sparse {SVD}
  calculations, version 2.1}.
\newblock \url{http://sun.stanford.edu/~rmunk/PROPACK/}, Apr. 2005.

\bibitem{MUSIC}
{\sc W.~Liao and A.~Fannjiang}, {\em Music for single-snapshot spectral
  estimation: {S}tability and super-resolution}, Applied and Computational
  Harmonic Analysis, 40 (2016), pp.~33--67.

\bibitem{MRI}
{\sc M.~Lustig, D.~Donoho, and J.~M. Pauly}, {\em Sparse {MRI}: The application
  of compressed sensing for rapid {MR} imaging}, Magnetic resonance in
  medicine, 58 (2007), pp.~1182--1195.

\bibitem{Radar}
{\sc L.~C. Potter, E.~Ertin, J.~T. Parker, and M.~Cetin}, {\em Sparsity and
  compressed sensing in radar imaging}, Proceedings of the IEEE, 98 (2010),
  pp.~1006--1020.

\bibitem{QMCCO:ACIE:15}
{\sc X.~Qu, M.~Mayzel, J.-F. Cai, Z.~Chen, and V.~Orekhov}, {\em Accelerated
  {NMR} spectroscopy with low-rank reconstruction}, Angewandte Chemie
  International Edition, 54 (2015), pp.~852--854.

\bibitem{P_Omega}
{\sc B.~Recht}, {\em A simpler approach to matrix completion}, The Journal of
  Machine Learning Research, 12 (2011), pp.~3413--3430.

\bibitem{rechtfazelparrilo2010nnm}
{\sc B.~Recht, M.~Fazel, and P.~A. Parrilo}, {\em Guaranteed minimum-rank
  solutions of linear matrix equations via nuclear norm minimization}, SIAM
  Review, 52 (2010), pp.~471--501.

\bibitem{Microscopy}
{\sc L.~Schermelleh, R.~Heintzmann, and H.~Leonhardt}, {\em A guide to
  super-resolution fluorescence microscopy}, The Journal of cell biology, 190
  (2010), pp.~165--175.

\bibitem{Tang}
{\sc G.~Tang, B.~N. Bhaskar, P.~Shah, and B.~Recht}, {\em Compressed sensing
  off the grid}, Information Theory, IEEE Transactions on, 59 (2013),
  pp.~7465--7490.

\bibitem{tw2012nihtmc}
{\sc J.~Tanner and K.~Wei}, {\em Normalized iterative hard thresholding for
  matrix completion}, SIAM Journal on Scientific Computing, 35 (2013),
  pp.~S104--S125.

\bibitem{Tropp}
{\sc J.~A. Tropp}, {\em User-friendly tail bounds for sums of random matrices},
  Foundations of computational mathematics, 12 (2012), pp.~389--434.

\bibitem{Analog}
{\sc J.~A. Tropp, J.~N. Laska, M.~F. Duarte, J.~K. Romberg, and R.~G.
  Baraniuk}, {\em Beyond nyquist: {E}fficient sampling of sparse bandlimited
  signals}, Information Theory, IEEE Transactions on, 56 (2010), pp.~520--544.

\bibitem{bart2012riemannian}
{\sc B.~Vandereycken}, {\em Low rank matrix completion by {R}iemannian
  optimization}, SIAM Journal on Optimization, 23 (2013), pp.~1214--1236.

\bibitem{Recovery}
{\sc K.~Wei, J.~F. Cai, T.~F. Chan, and S.~Leung}, {\em Guarantees of
  {R}iemannian optimization for low rank matrix recovery}, arXiv preprint
  arXiv:1511.01562,  (2015).

\bibitem{Completion}
\leavevmode\vrule height 2pt depth -1.6pt width 23pt, {\em Guarantees of
  {R}iemannian optimization for low rank matrix completion}, arXiv preprint
  arXiv:1603.06610,  (2016).

\bibitem{xu2008fast}
{\sc W.~Xu and S.~Qiao}, {\em A fast symmetric {SVD} algorithm for square
  {H}ankel matrices}, Linear Algebra and its Applications, 428 (2008),
  pp.~550--563.

\end{thebibliography}
\bibliographystyle{siam}



\end{document}